\documentclass[12pt,letterpaper]{article}
\usepackage[T1]{fontenc}
\usepackage{lmodern}
\usepackage[margin=1in]{geometry}
\usepackage{microtype}
\usepackage{amsmath,amssymb,amsthm,mathtools}
\usepackage{booktabs,array}
\usepackage{graphicx}
\usepackage{enumitem}
\usepackage{float}
\usepackage{titlesec}
\titleformat{\section}{\large\bfseries}{\thesection}{1em}{}
\titleformat{\subsection}{\normalsize\bfseries}{\thesubsection}{1em}{}
\usepackage[authoryear,round]{natbib}

\usepackage{url}
\usepackage{subcaption}
\usepackage[hidelinks]{hyperref}

\newtheorem{theorem}{Theorem}

\newtheorem{lemma}[theorem]{Lemma}
\newtheorem{remark}{Remark}
\newtheorem{assumption}{Assumption}
\newtheorem{corollary}[theorem]{Corollary}

\newcommand{\E}{\mathbb{E}}
\newcommand{\Cov}{\mathrm{Cov}}
\newcommand{\R}{\mathbb{R}}


\title{Beyond Point Predictions: Distribution-Informed Prediction-Powered Inference}
\author{
	{Maoyu Zhang$^{1}$,  Jingfei Zhang$^2$ and Xuming He$^{1}$}
	\vspace{1.6mm}\\
	\fontsize{11}{10}\selectfont\itshape $^1$\,Department of Statistics and Data Science, Washington University in St. Louis, St. Louis, MO, USA. \\
	\fontsize{11}{10}\selectfont\itshape $^2$\,Goizueta Business School, Emory University, Atlanta, GA, USA. }

\date{}

\begin{document}
\maketitle

\begin{abstract}

Prediction-powered inference (PPI) typically relies on point predictions on unlabeled data. When predictive distributions are available as in a wide range of applications, including predictions from LLMs, we introduce \textit{distribution-informed prediction-powered inference} (DiPPI), a general framework for further improving statistical efficiency by using predictive distributions as auxiliary information in the spirit of PPI. We characterize the optimal use of this information through score calibration, 
{derive the oracle efficiency for a finite-dimensional representation of the predictive distribution,} and provide theoretical guarantees for positive learning with the cross-fitted DiPPI estimator. Through simulations and three real-data applications, we show that DiPPI achieves better efficiency than PPI methods based on point predictions. These gains arise when the predictive distribution contains score-relevant information that is partially lost in point predictions.
\end{abstract}

\section{Introduction}
\label{sec:intro}

In many scientific {and social-economic} studies,  {reliable outcome labeling or measurements} often require costly experiments or human annotation. As a result, responses and covariates are jointly observed in a small labeled sample, whereas covariates alone are available for a much larger unlabeled sample. Pretrained machine learning models, including large language models (LLMs), can predict {or impute} responses from the available covariates and additional information, but treating these predictions as observed responses can lead to biased estimates and invalid confidence intervals. Prediction-powered inference \citep[PPI;][]{ppi} addresses this problem by combining predictions on the unlabeled sample with a correction estimated from the labeled sample, yielding valid inference on population parameters even when the predictions are biased. Using predictions in this way can improve efficiency relative to the estimator based on the labeled sample alone. {A rapidly growing literature has since developed adaptive weighting, recalibration and other strategies to further improve the efficiency of prediction-powered inference \citep{ppiplus,pspa,pdc,reppi}, among others.}

Existing PPI methods typically use a point prediction of the response to construct the correction and to improve statistical efficiency. However, prediction systems such as LLMs are designed to provide a predictive distribution for each case, that is, a probability distribution over the response space. For example, an LLM used to predict a product rating from review text may provide a probability distribution over the possible ratings. Reducing such a predictive distribution to a point prediction, such as its mean or mode, may discard information relevant to estimation of the population parameter of interest. This suggests shifting the focus of prediction-powered inference from point predictions to the richer information contained in predictive distributions. 

To illustrate, suppose we wish to estimate a school district coefficient in a linear regression of home sale price (per square foot) on school district and other property characteristics, where the coefficient captures the school district premium in property values. An LLM may use a property’s address to retrieve listing descriptions and recent sales records for comparable properties in the area, and then produce a predictive distribution for the property price. Two properties with the same observed covariates in the regression may have different predictive distributions because the retrievable information differs. If these distributional differences contain information about the expected sale price beyond the regression covariates and the point prediction, they can improve statistical efficiency within the PPI framework. %
{More broadly, this motivates a principled shift in PPI methodology from relying on point predictions for unlabeled data toward exploiting predictive distributions, or their low-dimensional summaries, as richer auxiliary information.}

In this paper, we propose distribution-informed prediction-powered inference (DiPPI), a {conceptually simple approach, in the spirit of PPI,} that uses predictive distributions as auxiliary information for inference on parameters defined by estimating equations. 
{Although the idea behind DiPPI is conceptually natural, we are not aware of any existing PPI variant that adopts it. One reason may be the lack of a straightforward way to exploit the predictive distribution as a whole.
To address this challenge, DiPPI builds on 
recalibrated prediction-powered inference \citep[RePPI;][]{reppi} %
and uses} the labeled sample to estimate the conditional expectation of the estimating function %
given the covariates and predictive distribution. %
The fitted conditional expectation is then evaluated for both labeled and unlabeled observations and used to construct the DiPPI estimating equation.

{The shift from point predictions to predictive distributions requires rigorous justification. We provide asymptotic results and numerical experiments to address when and why the proposed DiPPI yields meaningful efficiency gains over existing methods.} %
For the DiPPI estimator, we establish asymptotic normality and provide a consistent covariance estimator for confidence intervals. Under appropriate regularity conditions, DiPPI is asymptotically no less efficient than labeled-only estimation, even when the conditional expectation is not estimated consistently. When the conditional expectation is estimated consistently, DiPPI achieves the smallest asymptotic covariance among estimators within our estimating-equation framework. In other words, DiPPI shares two desirable asymptotic properties with RePPI: (i) it is never less efficient than the analysis based on the labeled sample alone; and (ii) if the conditional expectation is estimated consistently, it achieves the minimum asymptotic variance among all prediction-powered estimators in the class under consideration. Moreover, when the predictive distributions capture (some level of) heterogeneity in the uncertainty of the point predictions, DiPPI delivers additional efficiency gains over RePPI.

We evaluate DiPPI through simulations and three real-data applications: estimating the log-price coefficient for wine ratings~\citep{reppi}, slope coefficients in multinomial logistic regression of forest cover types~\citep{covertype}, and the hedging-indicator coefficient for politeness ratings~\citep{danescu2013politeness}. Wine and Politeness use predictive rating distributions from GPT-5.4 nano, while Forest uses class-probability predictions from a held-fixed XGBoost classifier~\citep{xgboost}. Across applications and labeled sample sizes, DiPPI produces shorter confidence intervals than PPI++ \citep{ppiplus}, RePPI \citep{reppi} and labeled-only estimation, with coverage close to the nominal level.

\subsection{Related Work}
\label{sec:related}
PPI \citep{ppi} combines model predictions on a large unlabeled sample with a correction estimated from a small labeled sample, following earlier work on post-prediction inference \citep{wang2020methods}. The literature on PPI is expanding rapidly, and we focus here on developments most closely related to our setting. 
Subsequent work improves efficiency through power tuning \citep{ppiplus}, prediction de-correlation \citep{pdc}, and adaptive weighting \citep{pspa}, post-hoc regression and calibration \citep{regressionmean,van2026calibeating}, empirical likelihood with auxiliary moment conditions \citep{epi}, among others. Closest to our work, RePPI~\citep{reppi} estimates the conditional expectation of the estimating function given covariates and a point prediction.
These methods typically use point predictions as auxiliary information and are closely related to classical control-variate methods and estimation with surrogate outcomes or partially observed responses \citep{robins1994,chenchen2000,chen2008,anotherlook}.

\section{Distribution-Informed Prediction-Powered Inference (DiPPI)}
\label{sec:framework}

Let $X\in\mathcal X\subseteq\mathbb R^p$ denote a covariate vector and $Y\in\mathcal Y\subseteq\mathbb R$ a response. For a known estimating function $U_\theta:\mathcal X\times\mathcal Y\to\mathbb R^d$, the parameter of interest $\theta^\star\in\Theta\subseteq\mathbb R^d$ satisfies
$$
\mathbb E[U_{\theta^\star}(X,Y)]=0,
$$
where the expectation is taken under the joint distribution of $(X,Y)$. For example, $U_{\theta^\star}(X,Y)$ may be a score function $\nabla_\theta\log p_\theta(Y|X)$ for a conditional probability mass function or density $p_\theta$.

A fixed prediction system $\mathcal A$ uses $(X,Z)$ to produce a predictive distribution $Q=\mathcal A(X,Z)\in\mathcal P(\mathcal Y)$, where $\mathcal P(\mathcal Y)$ denotes the space of probability distributions on $\mathcal Y$ and $Z$ denotes additional supplied or
retrieved inputs. For example, $Q$ may specify probabilities over possible 1-10 ratings or a predictive distribution for a continuous response. A point prediction, such as the mean or mode of $Q$, summarizes this distribution but may discard information relevant to the target estimating equation. Our framework uses $Q$ as auxiliary information
and does not require it to equal the conditional distribution of $Y$ given $(X,Z)$.

We observe a labeled sample of size $n$ and an unlabeled sample of size $N$,
\[
    \mathcal D_L
    =
    \{(X_i,Y_i,Q_i):i=1,\ldots,n\},
    \qquad
    \mathcal D_U
    =
    \{(X_j,Q_j):j=n+1,\ldots,n+N\},
\]
where $Q_i=\mathcal A(X_i,Z_i)$. 
We assume that the labeled observations are i.i.d. from the joint distribution of $(X,Y,Q)$, that the unlabeled observations are i.i.d. from the joint distribution of $(X,Q)$, and that the two samples are independent. Thus, both samples share the same marginal distribution of $(X,Q)$, while responses are observed only in $\mathcal D_L$.

\begin{remark}
The inputs $Z$ may be available to us but difficult to model directly, or accessible only to the prediction system. For example, in a regression of wine ratings, $Z$ may contain a textual wine description supplied to an LLM. In a regression of housing price, $Z$ may contain listing descriptions and historical property records retrieved using an address not included in $X$. Our procedure uses $Z$ only through $Q$, without requiring direct access to or modeling of $Z$.
\end{remark}

\subsection{Distribution-informed estimation}
Estimating $\theta^\star$ requires estimating $\mathbb E\{U_\theta(X,Y)\}$ as
a function of $\theta$. To use the unlabeled sample for this purpose, we use the predictive distribution $Q$ and $X$ to predict $U_\theta(X,Y)$.  
Let $h_\theta:\mathcal X\times\mathcal P(\mathcal Y)\to\mathbb R^d$ be a fixed prediction function. We write
$$
\mathbb E[U_\theta(X,Y)]=\mathbb E[U_\theta(X,Y)-h_\theta(X,Q)]+\mathbb E[h_\theta(X,Q)].
$$
Since $(X,Q)$ is available in both samples, the second term can be estimated from all $n+N$ observations, whereas the first requires $Y$ and can be estimated only from the $n$ labeled observations. We therefore estimate $\theta^\star$ by solving
\begin{equation}\label{eq:ee}
\Psi_{n,N}(\theta;h)=\frac{1}{n}\sum_{i=1}^{n}\{U_\theta(X_i,Y_i)-h_\theta(X_i,Q_i)\}+
\frac{1}{n+N}\sum_{k=1}^{n+N}h_\theta(X_k,Q_k)=0.
\end{equation}
Under our sampling assumptions, the left-hand side of \eqref{eq:ee} has expectation $\mathbb E\{U_\theta(X,Y)\}$ for every $h_\theta$, without requiring $Q$ to be the conditional distribution of $Y$ given $(X,Z)$. Taking $h_\theta=0$ recovers the labeled-only estimating equation.

The choice of $h_\theta$ affects the efficiency of the resulting estimator, denoted as $\widehat\theta_h$.
The following result gives the asymptotic covariance of $\widehat\theta_h$ and identifies the optimal $h_\theta$ that minimizes this covariance under two regularity conditions. 

\begin{assumption}
\label{assump:basic}
There exists $\rho>0$ such that
$
\mathcal B_\rho
:=
\{\theta\in\mathbb R^d:\|\theta-\theta^\star\|\leq\rho\}
\subset\operatorname{int}(\Theta).
$
The function $U_\theta(X,Y)$ is continuously differentiable in $\theta$ on
$\mathcal B_\rho$. Moreover,
$
H_{\theta^\star}
=
\E\{\nabla_\theta U_{\theta^\star}(X,Y)\}
$
is nonsingular, and
$
\E\!\left[
\|U_{\theta^\star}(X,Y)\|^2
+
\sup_{\theta\in\mathcal B_\rho}
\|\nabla_\theta U_\theta(X,Y)\|
\right]
<\infty.
$
\end{assumption}

\begin{assumption}
\label{assump:fixed-correction}
The function $h_\theta(X,Q)$ is continuously differentiable in $\theta$ on
$\mathcal B_\rho$, and
$
\E\!\left[
\|h_{\theta^\star}(X,Q)\|^2
+
\sup_{\theta\in\mathcal B_\rho}
\|\nabla_\theta h_\theta(X,Q)\|
\right]
<\infty.
$
\end{assumption}

\begin{theorem}
\label{thm:fixed-correction}
Suppose Assumptions \ref{assump:basic}-\ref{assump:fixed-correction} hold and $n/N\to r\in(0,\infty)$. Then there
exists $a\in(0,\rho)$ such that, with probability tending to one,
$\Psi_{n,N}(\theta;h)=0$ has a unique root $\widehat\theta_h$ in
$\mathcal B_a=\{\theta:\|\theta-\theta^\star\|\leq a\}$, 
and this root belongs to $\operatorname{int}(\mathcal B_a)$. Moreover,
\begin{equation}
\begin{gathered}
\sqrt n(\widehat\theta_h-\theta^\star)
\xrightarrow{d}
N(0,\Sigma_h),\\
\Sigma_h=
H_{\theta^\star}^{-1}
\left[
\frac{1}{1+r}\Cov\{U_{\theta^\star}(X,Y)-h_{\theta^\star}(X,Q)\}+
\frac{r}{1+r}\Cov\{U_{\theta^\star}(X,Y)\}
\right]
H_{\theta^\star}^{-\top}.
\end{gathered}
\label{eq:general-clt}
\end{equation}
Among all $h_{\theta}$ satisfying Assumption \ref{assump:fixed-correction}, $\Sigma_h$ is minimized by any $h^\star$ satisfying
\begin{equation}
h^\star_{\theta^\star}(X,Q)
=\E\{U_{\theta^\star}(X,Y)\mid X,Q\}.
\label{eq:oracle-h-Q}
\end{equation}
The corresponding minimum covariance is
\begin{equation}
\begin{split}
\Sigma_{h^\star}
={}&
H_{\theta^\star}^{-1}
\Bigg[
\Cov\{U_{\theta^\star}(X,Y)\}-
\frac{1}{1+r}
\Cov\!\left\{
\E[U_{\theta^\star}(X,Y)\mid X,Q]
\right\}
\Bigg]
H_{\theta^\star}^{-\top}.
\end{split}
\label{eq:oracle-covariance}
\end{equation}
\end{theorem}
Theorem~\ref{thm:fixed-correction} shows that using $Q$ can improve efficiency when
$\mathbb E\{U_{\theta^\star}(X,Y)\mid X,Q\}$ differs from
$\mathbb E\{U_{\theta^\star}(X,Y)\mid X\}$.
This conditional expectation may depend on features of $Q$ beyond its mean, such as shape, variance or tail probabilities. Using these features can therefore improve prediction of the estimating function and reduce the asymptotic covariance. 
We discuss the estimation of $\mathbb E\{U_{\theta^\star}(X,Y)\mid X,Q\}$ in Section \ref{sec:dippi-estimation}.
Taking $h_\theta=0$ recovers labeled-only estimation, while $h_\theta(X,Q)=(1+n/N)U_\theta(X,f(X))$ recovers PPI with point prediction $f(X)$.
Thus, the optimal $h^\star_{\theta^\star}(X,Q)$ in Theorem~\ref{thm:fixed-correction} yields asymptotic covariance no greater than that of either estimator.

\subsection{Beyond point predictions}
\label{sec:distributional-information}

Let $\widehat Y=T(Q)$ denote a point prediction obtained from $Q$ through a summary function $T$. We study when using $Q$ improves asymptotic efficiency. A direct way to predict $U_\theta(X,y)$ using $Q$ is
\[
s_\theta(X,Q)=\int U_\theta(X,y)\,Q(dy).
\]
When $s_\theta$ satisfies Assumption~\ref{assump:fixed-correction}, it can be used as
$h_\theta$ in \eqref{eq:ee}. However, $s_{\theta^\star}(X,Q)$ need not equal the optimal predictor $h^\star_{\theta^\star}(X,Q)=\mathbb E[U_{\theta^\star}(X,Y)\mid X,Q]$. 

For mean estimation and the usual linear and logistic regression
estimating functions,
\[
s_{\theta^\star}(X,Q)=U_{\theta^\star}(X,m(Q)),
\qquad
h^\star_{\theta^\star}(X,Q)
=
U_{\theta^\star}\!\left(X,\mathbb E[Y\mid X,Q]\right),
\]
where $m(Q)=\int y\,Q(dy)$. The two response predictions differ by
\[
\mathbb E[Y\mid X,Q]-m(Q)
=
\mathbb E[Y-m(Q)\mid X,Q],
\]
the conditional mean prediction error of $m(Q)$. Features such as the variance, skewness, or tail probabilities of $Q$ may help predict this error conditional on $(X,m(Q))$.  If $\mathbb E[Y|X,Q]=\mathbb E[Y|X,m(Q)]$, then the optimal prediction of $U_\theta(X,Y)$ based on $(X,Q)$ is also a function of $(X,m(Q))$. In this case, having access to the full predictive distribution gives the same minimum asymptotic covariance as having access to only its mean.

For other estimating functions, $s_\theta$ may depend on features of $Q$ beyond its mean. In multinomial logistic regression with class $K$ as the reference category, let $p_{\theta,k}(X)$ and $q_k$ denote the model and predicted probabilities for class $k$. For $k=1,\ldots,K-1$,
$$
U_{\theta,k}(X,Y)=X\{\mathbf 1(Y=k)-p_{\theta,k}(X)\},\qquad
s_{\theta,k}(X,Q)=X\{q_k-p_{\theta,k}(X)\}.
$$
Hence, $s_\theta$ depends on the class probabilities, which are generally not determined by the predictive mean or modal class; reducing $Q$ to either summary can therefore lose relevant information.

We now compare the minimum asymptotic covariances achievable with $Q$ and $\widehat Y=T(Q)$. Define
\[
h^{\star,T}_{\theta^\star}(X,\widehat Y)
=
\mathbb E[U_{\theta^\star}(X,Y)\mid X,\widehat Y].
\]
By Theorem~\ref{thm:fixed-correction}, this predictor minimizes
the asymptotic covariance among prediction functions based on
$(X,\widehat Y)$. Denote this minimum by $\Sigma_T^\star$.
The corresponding minimum based on $(X,Q)$ is
$\Sigma_{h^\star}$. The law of total covariance gives
\begin{equation}
\Sigma_T^\star-\Sigma_{h^\star}
=
\frac{1}{1+r}H_{\theta^\star}^{-1}
\Cov\!\left\{
h^\star_{\theta^\star}(X,Q)
-
h^{\star,T}_{\theta^\star}(X,\widehat Y)
\right\}
H_{\theta^\star}^{-\top}
\succeq0.
\label{eq:distribution-point-gap}
\end{equation}
Thus, replacing $Q$ by $\widehat Y$ gives no loss of efficiency
if and only if
\[
\mathbb E[U_{\theta^\star}(X,Y)\mid X,Q]
=
\mathbb E[U_{\theta^\star}(X,Y)\mid X,\widehat Y].
\]
The predictor $h^{\star,T}_{\theta^\star}$ is the conditional expectation targeted by RePPI~\citep{reppi}. Thus, \eqref{eq:distribution-point-gap} shows that retaining $Q$
can improve efficiency even relative to a point-based method that consistently estimates this conditional expectation. This comparison assumes that $\widehat Y=T(Q)$ and does not establish an efficiency ordering for separately generated predictions that are not determined by $(X,Q)$.

\subsection{Cross-fitted estimation and matrix calibration}
\label{sec:dippi-estimation}

To use $Q$ as a regression input, we choose a finite-dimensional representation $\Phi(Q)\in\mathbb R^q$. For a discrete response with $K$ possible values, $\Phi(Q)$ can consist of $K-1$ class probabilities. For a continuous response, the choice depends on how the prediction system provides $Q$. An available density or CDF can be represented by basis coefficients or functional principal component scores. If only generated responses are available, these functions must be approximated from repeated samples. In this case, selected moments or quantiles may offer a simpler representation; a practical starting point for $\Phi(Q)$ is the predictive mean, standard deviation, and a few quantiles.

We then estimate $\E\{U_{\theta^\star}(X,Y)\mid X,\Phi(Q)\}$ from the labeled data using three-fold cross-fitting. Because $\theta^\star$ is unknown, one fold provides an initial estimator $\widehat\theta_0$. A second fold is used to construct $\widehat\mu(X,Q)\approx \E\{U_{\widehat\theta_0}(X,Y)\mid X,\Phi(Q)\}$. 
Since an inaccurate estimate of the conditional score can increase variance, we calibrate $\widehat\mu(X,Q)$ by a matrix $M$ chosen to minimize
the asymptotic covariance. This step is related to optimal control-variate and
power-tuning constructions \citep{chenchen2000,ppiplus,pdc,anotherlook}. 
Replacing $\theta^\star$ by $\widehat\theta_0$ and the population covariances in this minimizer by empirical covariances on a third labeled fold gives
\begin{equation}
    \widehat M
    =
    \widehat{\Cov}\!\left\{
    U_{\widehat\theta_0}(X,Y),\widehat\mu(X,Q)
    \right\}
    \widehat{\Cov}\!\left\{\widehat\mu(X,Q)\right\}^{-1}.
    \label{eq:M-hat}
\end{equation}
We refer to this step as \emph{matrix calibration}. If the empirical covariance is singular, its generalized inverse may be used. {{The final estimator is
$\widehat h(X,Q)=\widehat M\widehat\mu(X,Q)$.}}

The procedure is summarized in \hyperlink{alg:dippi}{Algorithm~1}.
We split the labeled sample into three folds $\mathcal D_1,\mathcal D_2,\mathcal D_3$. For each $k\in\{1,2,3\}$, $\mathcal D_k$ is used for calibration
and inference, while the other two folds separately provide
$\widehat\theta_0^k$ and $\widehat\mu^k$.
Hence, $(\widehat\theta_0^k,\widehat\mu^k)$ is independent of $\mathcal D_k$ and the unlabeled sample. We combine the resulting estimates with weights $\omega_k=|\mathcal D_k|/n$. 

\begin{figure}[ht!]
\centering
\begin{minipage}{0.96\linewidth}
\hrule
\vspace{0.4em}
\hypertarget{alg:dippi}{\textbf{Algorithm 1: Distribution-Informed Prediction-Powered Inference (DiPPI)}}
\vspace{0.4em}
\hrule
\vspace{0.5em}

\textbf{Step 1:} Split the labeled sample into three folds $\mathcal D_1$, $\mathcal D_2$, and $\mathcal D_3$.

\textbf{Step 2:} On $\mathcal D_2$, compute $\widehat\theta_0^{1}$ as the
{labeled-only local root} of
$|\mathcal D_2|^{-1}\sum_{i\in\mathcal D_2}U_\theta(X_i,Y_i)=0$.

\textbf{Step 3:} On $\mathcal D_3$, fit $\widehat\mu^{1}(X,Q)\approx\E\{U_{\widehat\theta_0^{1}}(X,Y)\mid X,\Phi(Q)\}$.

\textbf{Step 4:} On $\mathcal D_1$, compute $\widehat M^{1}$ using~\eqref{eq:M-hat}.

\textbf{Step 5:} Using $\mathcal D_1$ and the full unlabeled sample, solve
\[
0=
\frac{1}{|\mathcal D_1|}\sum_{i\in\mathcal D_1}
\Big[
U_\theta(X_i,Y_i)
-\frac{1}{1+n/N}\widehat M^{1}\widehat\mu^{1}(X_i,Q_i)
\Big]
+
\frac{1}{N}\sum_{j=n+1}^{n+N}
\frac{1}{1+n/N}\widehat M^{1}\widehat\mu^{1}(X_j,Q_j),
\]
and denote the solution by $\widehat\theta^{1}$.

\textbf{Step 6:} Repeat Steps 2--5 with $(\mathcal D_3,\mathcal D_1,\mathcal D_2)$ and $(\mathcal D_1,\mathcal D_2,\mathcal D_3)$; obtain $\widehat\theta^{2}$ and $\widehat\theta^{3}$.

\textbf{Step 7:} Return $\widehat\theta^{\mathrm{DiPPI}}=\sum_{k=1}^3\omega_k\widehat\theta^k$.
\vspace{0.2em}
\hrule
\end{minipage}
\end{figure}

In Step 3, $\widehat\mu(X,Q)$ may be estimated using any suitable regression method with input $(X,\Phi(Q))$, including linear, regularized, additive, tree-based, or neural-network models. The choice should reflect the labeled sample size and the complexity and dimension of the regression problem. Model selection and tuning are performed within the fold used to fit $\widehat\mu$.

\subsection{Theoretical guarantees and inference}
\label{sec:dippi-theory}

We establish asymptotic guarantees for the cross-fitted estimator $\widehat\theta^{\mathrm{DiPPI}}$ in \hyperlink{alg:dippi}{Algorithm~1}, with estimation
of the prediction functions and calibration matrices. We also establish consistency of the covariance estimator and asymptotic validity of the resulting Wald confidence intervals.

{
\begin{assumption}
\label{assump:nuisance}
For each $k=1,2,3$, the fold size satisfies
$n_k/n\to\pi_k\in(0,1)$, and
$
\E\!\left[
\sup_{\theta\in\mathcal B_\rho}
\|\nabla_\theta U_\theta(X,Y)\|^2
\right]<\infty.
$
\end{assumption}
}

\begin{theorem}
\label{thm:feasible-dippi}
{Suppose Assumptions~\ref{assump:basic}
and~\ref{assump:nuisance} hold and $n/N\to r\in(0,\infty)$. If, for some
square-integrable $\mu(X,Q)$ with
$\Cov\{\mu(X,Q)\}\succ0$,
$
\max_{1\leq k\leq3}
\E\!\left\{\|\widehat\mu^k(X,Q)-\mu(X,Q)\|^2\right\}=o_p(1),
$
then}
$\widehat\theta^{\mathrm{DiPPI}}\xrightarrow{p}\theta^\star$ and
$
    \sqrt n(\widehat\theta^{\mathrm{DiPPI}}-\theta^\star)
    \xrightarrow{d}
    N\left(0,\Sigma_\mu^{\mathrm{DiPPI}}\right),
$
where
\begin{equation}
\begin{aligned}
    \Sigma_\mu^{\mathrm{DiPPI}}
    =H_{\theta^\star}^{-1}\Big[&\Cov\{U_{\theta^\star}(X,Y)\}
    -\frac{1}{1+r}
    \Cov\{U_{\theta^\star}(X,Y),\mu(X,Q)\}
    \Cov\{\mu(X,Q)\}^{-1}\\
    &\times
    \Cov\{\mu(X,Q),U_{\theta^\star}(X,Y)\}
    \Big]H_{\theta^\star}^{-\top}.
\end{aligned}
\label{eq:dippi-covariance}
\end{equation}
\end{theorem}

\begin{corollary}
\label{cor:dippi-safety}
Under the conditions of Theorem~\ref{thm:feasible-dippi}, let $\Sigma^{\mathrm C}=H_{\theta^\star}^{-1}\Cov\{U_{\theta^\star}(X,Y)\}H_{\theta^\star}^{-\top}$ denote the labeled-only asymptotic covariance. Then
\begin{equation}
\begin{aligned}
    \Sigma^{\mathrm C}-\Sigma_\mu^{\mathrm{DiPPI}}
    ={}&\frac{1}{1+r}H_{\theta^\star}^{-1}
    \Cov\{U_{\theta^\star}(X,Y),\mu(X,Q)\}
    \Cov\{\mu(X,Q)\}^{-1}\\
    &\quad\times
    \Cov\{\mu(X,Q),U_{\theta^\star}(X,Y)\}
    H_{\theta^\star}^{-\top}
    \succeq0.
\end{aligned}
    \label{eq:dippi-safety-gap}
\end{equation}
Thus, DiPPI is asymptotically no less efficient than labeled-only inference, even when $\mu(X,Q)$ differs from $\mathbb E[U_{\theta^\star}(X,Y)\mid X,\Phi(Q)]$.
If, in addition, $\mu(X,Q)=E\{U_{\theta^\star}(X,Y)\mid X,\Phi(Q)\}$, then
\begin{equation}
\Sigma_\mu^{\mathrm{DiPPI}}
=
H_{\theta^\star}^{-1}
\Big[
\Cov\{U_{\theta^\star}(X,Y)\}
-
\frac{1}{1+r}
\Cov\!\left\{
\E[U_{\theta^\star}(X,Y)\mid X,\Phi(Q)]
\right\}
\Big]
H_{\theta^\star}^{-\top}.
\label{eq:dippi-oracle-covariance}
\end{equation}
\end{corollary}
Corollary~\ref{cor:dippi-safety} shows that
$\Sigma_\mu^{\mathrm{DiPPI}}\preceq\Sigma^{\mathrm C}$ even if
$\mu(X,Q)\neq\E\{U_{\theta^\star}(X,Y)\mid X,\Phi(Q)\}$.
When $\mu(X,Q)=\E\{U_{\theta^\star}(X,Y)\mid X,\Phi(Q)\}$, DiPPI attains the minimum asymptotic covariance
among prediction functions based on $(X,\Phi(Q))$ satisfying
Assumption~\ref{assump:fixed-correction}.

For variance estimation, let $\widehat{\Cov}_{\mathcal D_k}$ and $\widehat{\Cov}_{\mathrm U}$ denote empirical covariances over $\mathcal D_k$ and the unlabeled sample, respectively. Because $\omega_k/|\mathcal D_k|=1/n$, each labeled observation receives weight $1/n$ after the fold estimates are averaged. All folds use the same unlabeled observations, so their unlabeled corrections are averaged before computing the covariance. We estimate the asymptotic covariance by
\begin{equation}
\begin{aligned}
\widehat\Sigma^{\mathrm{DiPPI}}
=\widehat H^{-1}\Bigg[&
\sum_{k=1}^3\omega_k
\widehat{\Cov}_{\mathcal D_k}\!\left\{
U_{\widehat\theta^{\mathrm{DiPPI}}}(X_i,Y_i)
   -\frac{\widehat M^k\widehat\mu^k(X_i,Q_i)}{1+n/N}
\right\}\\
&+\frac{n}{N}\widehat{\Cov}_{\mathrm U}\!\left\{
\frac{1}{1+n/N}\sum_{k=1}^3\omega_k
\widehat M^k\widehat\mu^k(X_j,Q_j)
\right\}\Bigg]\widehat H^{-\top},
\end{aligned}
\label{eq:dippi-sandwich-estimator}
\end{equation}
For smooth scores, one may take $\widehat H=n^{-1}\sum_{i=1}^n\nabla_\theta U_{\widehat\theta^{\mathrm{DiPPI}}}(X_i,Y_i)$. More generally, we only require $\widehat H\xrightarrow{p}H_{\theta^\star}$. The following theorem gives covariance consistency.

\begin{theorem}
\label{thm:dippi-variance}
{Suppose the conditions of Theorem~\ref{thm:feasible-dippi}
hold and $\widehat H\xrightarrow{p}H_{\theta^\star}$. Then}
\begin{equation}
    \widehat\Sigma^{\mathrm{DiPPI}}
    \xrightarrow{p}
    \Sigma_\mu^{\mathrm{DiPPI}}.
    \label{eq:dippi-variance-consistency}
\end{equation}
For $a\in\R^d$ satisfying
$a^\top\Sigma_\mu^{\mathrm{DiPPI}}a>0$,
\begin{equation}
    \frac{\sqrt n\,a^\top
    (\widehat\theta^{\mathrm{DiPPI}}-\theta^\star)}
    {\{a^\top\widehat\Sigma^{\mathrm{DiPPI}}a\}^{1/2}}
    \xrightarrow{d}N(0,1),
    \label{eq:dippi-studentized-clt}
\end{equation}
and the Wald interval
$a^\top\widehat\theta^{\mathrm{DiPPI}}\pm z_{1-\alpha/2}(\frac{a^\top\widehat\Sigma^{\mathrm{DiPPI}}a}{n})^{1/2}$
has asymptotic coverage $1-\alpha$ for $a^\top\theta^\star$. %
\end{theorem}

\section{Simulation Studies}
\label{sec:sim}

We consider two simulation settings in which predictive variance provides information beyond the predictive mean. In the first setting, observations with the same covariates and predictive mean can have different expected responses, and the predictive variance provides information about these differences.
In the second setting, predictive variance indicates which point predictions are more reliable, while the overall prediction bias and mean squared error remain fixed. 
Additional simulations are reported in Appendix~\ref{app:additional-simulation}.

We compare labeled-only estimation ({Classical}), PPI++, RePPI, and DiPPI. We specify $Q$ directly, with PPI++ and RePPI using its mean $\widehat Y=\E_Q(Y)$ and DiPPI using the distributional features $\Phi(Q)$ defined below:
\[
\Phi(Q)=
\bigl(
\operatorname{mean}(Q),
\operatorname{sd}(Q),
q_{0.1}(Q),
q_{0.25}(Q),
q_{0.5}(Q),
q_{0.75}(Q),
q_{0.9}(Q)
\bigr).
\]
We estimate the conditional mean of $Y$ using linear regression in Experiment 1 and histogram gradient boosting in Experiment 2. Both experiments use $N=10{,}000$ and $200$ replications.
We report relative RMSE to Classical, mean $95\%$ confidence-interval width, and
empirical coverage. For method $m$,
$
\operatorname{relative\ RMSE}_m
=
\left\{
{\sum_{b=1}^R\|\widehat\theta_m^{(b)}-\theta\|_2^2}/
{\sum_{b=1}^R\|\widehat\theta_{\mathrm C}^{(b)}-\theta\|_2^2}
\right\}^{1/2},
$
where $\widehat\theta_{\mathrm C}^{(b)}$ denotes the Classical estimator.
Classical therefore has relative RMSE one, with smaller values indicating greater
accuracy. We report confidence-interval width and coverage averaged across the slope coefficients.

\textbf{Experiment 1: Outcome-relevant predictive variance.}
We examine whether predictive variance can improve inference beyond the predictive mean. Let $X,Z\in\mathbb R^5$ have independent standard-normal
coordinates, let $V\sim\operatorname{Unif}(1,4)$, and let
$\varepsilon\sim N(0,1)$, independently. We generate
\[
Y
=
X^\top\theta
+
Z^\top\beta
+
\kappa(V-2.5)
+
\varepsilon,
\qquad
Q=N(Z^\top\beta,V),
\qquad
\widehat Y=Z^\top\beta.
\]
Here, $V$ affects the conditional mean of $Y$ but appears only in the variance of $Q$. Thus, $Q$ provides information about $Y$ beyond $(X,\widehat Y)$ when $\kappa\neq0$. We estimate the slope coefficients $\theta$ with $n=1,000$ and $N=10,000$, varying $\kappa\in\{0,0.5,1,1.5,2,2.5,3\}$. Figure~\ref{fig:sim-spread} shows similar performance for DiPPI and RePPI when $\kappa=0$. As $\kappa$ increases, DiPPI achieves larger efficiency gains over RePPI, with coverage remaining close to nominal.

\begin{figure}[!htbp]
    \centering
    \begin{subfigure}{0.98\linewidth}
        \caption{Experiment 1: Outcome-relevant predictive variance}\label{fig:sim-spread}
        \includegraphics[width=\linewidth]{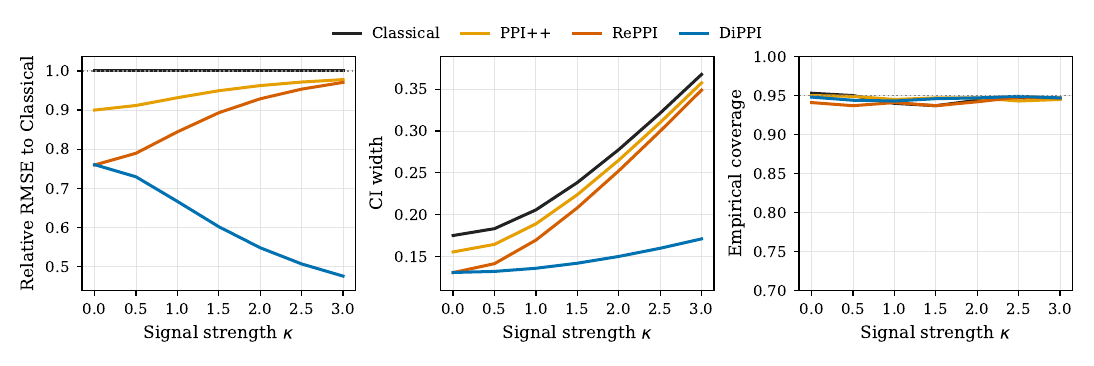}
    \end{subfigure}
    \par\medskip
    \begin{subfigure}{0.98\linewidth}
        \caption{Experiment 2: Heterogeneous prediction-error variances}\label{fig:sim-reliability}
        \includegraphics[width=\linewidth]{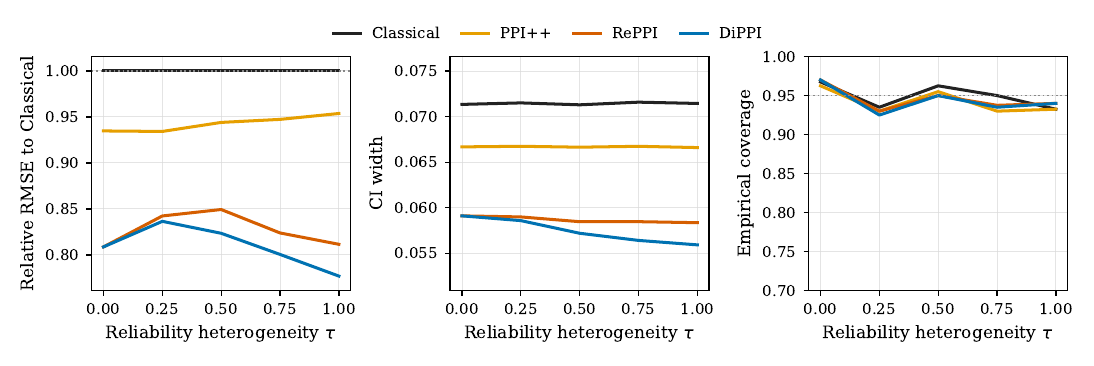}
    \end{subfigure}
    \par\medskip
    \caption{Relative RMSE to Classical, mean $95\%$ CI width, and empirical coverage.}
    \label{fig:simulations}
\end{figure}

\textbf{Experiment 2: Heterogeneous prediction-error variances.}

We next consider point predictions whose accuracy differs across observations. Generate
\[
X_i\sim N(0,I_2),
\qquad
Y_i=X_i^\top\theta^\star+\varepsilon_i,
\qquad
\theta^\star=(1,-1)^\top/\sqrt{2},
\qquad
\varepsilon_i\sim N(0,1).
\]
Independently, draw $A_i$ uniformly from $\{-1,1\}$ and set
\[
S_i
=
\frac{\exp(\tau A_i)}{\sqrt{\cosh(2\tau)}},
\qquad
\xi_i\mid S_i\sim N(-2,S_i^2),
\qquad
\widehat Y_i=Y_i+\xi_i,
\qquad
Q_i=N(\widehat Y_i,S_i^2).
\]
The parameter $\tau\in\{0,0.25,0.5,0.75,1\}$ controls the variation in prediction-error variance across observations. The normalization gives
$
\E(S_i^2)=1$ and $
\E\{(\widehat Y_i-Y_i)^2\}=5,
$
so the prediction bias and mean squared error remain fixed as $\tau$ varies.  
At $\tau=0$, $S_i=1$, and $Q_i$ provides no information beyond $\widehat Y_i$. For $\tau>0$, its variance indicates how accurately $\widehat Y_i$ measures $Y_i$. This information helps estimate the conditional mean of $Y_i$ given $X_i$ and $\widehat Y_i$, since predictions with smaller error variances receive greater weight. We estimate the regression slopes $\theta^\star$ with $n=3{,}000$ and $N=10{,}000$. Figure~\ref{fig:sim-reliability} shows similar performance for DiPPI and RePPI at $\tau=0$, with increasing efficiency gains for DiPPI as $\tau$ increases.

\section{Real-Data Experiments}
\label{sec:realdata}

We evaluate DiPPI using Wine, Forest, and Politeness data. Wine and Politeness examine whether predictive distributions improve inference beyond their means, while Forest examines whether class probabilities improve inference beyond point predictions in multinomial regression. 

We compare {Classical}, {PPI++}, {RePPI}, and {DiPPI} using point predictions derived from the same predictive distribution $Q$ used by DiPPI. DiPPI and RePPI use identical data folds and histogram gradient boosting to estimate conditional scores. For each labeled ratio $r$, we independently draw $n=rN$ labeled and $N$ unlabeled observations with replacement from the analysis population and repeat the experiment $200$ times. {Using the full-population estimate as the ground-truth target, we report relative RMSE, mean $95\%$ CI width, and empirical coverage, as in Section~\ref{sec:sim}}.

\subsection{Wine ratings}
\label{sec:real-wine}
The Wine data
\footnote{\url{https://www.kaggle.com/datasets/mysarahmadbhat/wine-tasting}}
contain $10,000$ U.S. reviews with ratings from 80 to 99. We regress rating on log price and indicators for
California, Washington, and Oregon, and target the log-price coefficient. We set $r\in\{0.05,0.10,0.15,0.20,0.25\}$.
{We use GPT-5.4 nano to construct a predictive distribution over all 20 possible ratings from the written review.} For each of PPI++ and RePPI, we evaluate two variants for point predictions, one using the predictive mean $\widehat Y=\E_Q(Y)$
and the other using the modal rating
$\widehat Y=\arg\max_{y\in\{80,\ldots,99\}}Q(y)$, whereas DiPPI uses the full
probability vector. This setting tests whether features of the rating
distribution beyond its mean or mode, such as dispersion or asymmetry, contain additional
information about the OLS score. Figure~\ref{fig:wine} shows that DiPPI has the lowest relative RMSE and
shortest intervals across all labeled ratios, with coverage near the nominal
level.
\begin{figure}[!htbp]
    \centering
    \begin{subfigure}{0.98\linewidth}
        \caption{Wine Data}\label{fig:wine}
        \includegraphics[width=\linewidth]{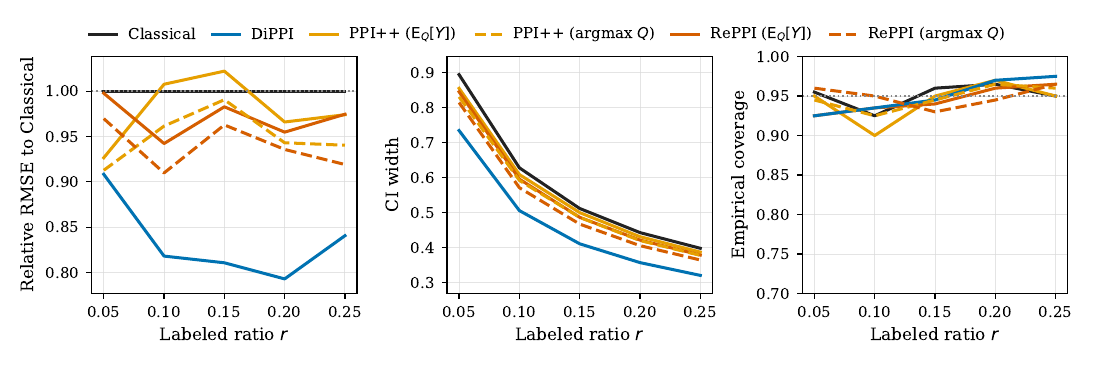}
    \end{subfigure}
    \par\medskip
    \begin{subfigure}{0.98\linewidth}
        \caption{Forest Data}\label{fig:forest}
        \includegraphics[width=\linewidth]{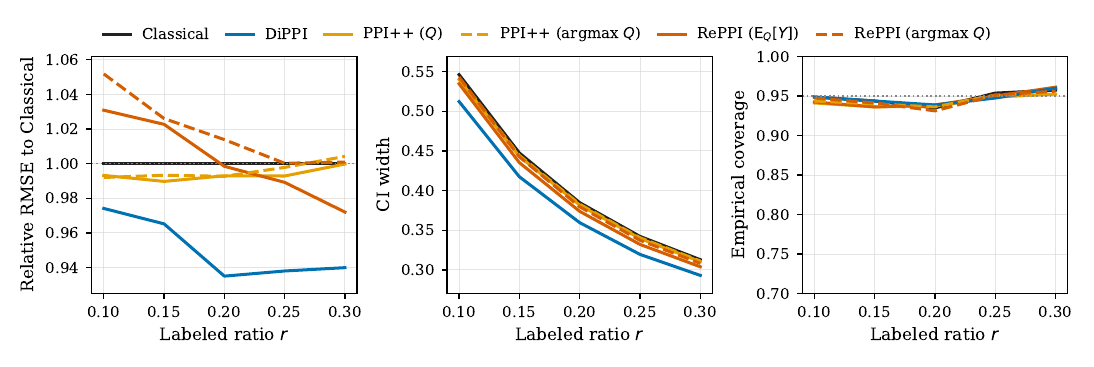}
    \end{subfigure}
    \par\medskip
    \begin{subfigure}{0.98\linewidth}
        \caption{Politeness Data}\label{fig:politeness}
        \includegraphics[width=\linewidth]{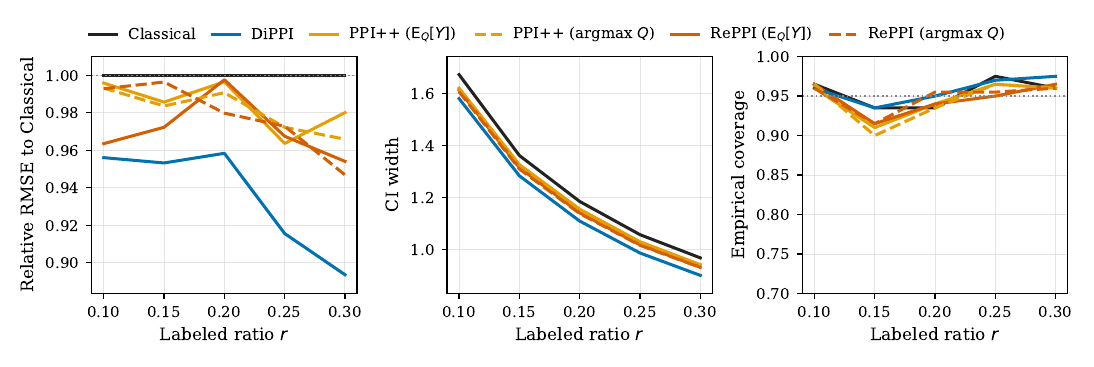}
    \end{subfigure}
    \par\medskip
    \caption{Relative RMSE to Classical, mean $95\%$ CI width, and empirical coverage.}
    \label{fig:realdata}
\end{figure}

\subsection{Forest cover types}
\label{sec:real-forest}
The Forest CoverType data\footnote{\url{https://archive.ics.uci.edu/dataset/31/covertype}} contain forest cover labels and cartographic covariates
from four wilderness areas in northern Colorado~\citep{covertype}. We retain
spruce/fir, lodgepole pine, ponderosa pine, and krummholz as separate classes,
and combine cottonwood/willow, aspen, and Douglas-fir into a fifth class, Rare.
The inferential target is the four class-specific coefficients of standardized
slope in a five-class multinomial-logit regression with an intercept and
standardized noon hillshade, using Rare as the reference class.

We train an XGBoost classifier~\citep{xgboost} on 15,000 observations sampled
from Rawah, Neota, and Cache la Poudre, using all 54 covariates. The analysis
population consists of 40,000 observations sampled from Comanche Peak, which
contains all five classes. The fitted classifier is held fixed and returns a
class-probability vector $Q$ for each observation. We set $r\in\{0.10,0.15,0.20,0.25,0.30\}$.
PPI++ uses $Q$ directly in the averaged multinomial loss.
RePPI uses 
$\widehat Y=\E_Q(Y)$, whereas DiPPI uses the full vector $Q$. We also evaluate PPI++ and RePPI using the modal class
$\widehat Y=\arg\max_{k\in\{1,\ldots,5\}}Q_k$
as the point prediction.
Figure~\ref{fig:forest} shows that DiPPI has the lowest relative RMSE and
shortest intervals across all labeled ratios, with coverage near the nominal
level.

\subsection{Politeness requests}
\label{sec:real-politeness}

The Politeness data contain Stack Exchange and Wikipedia requests with human
ratings on a 1--25 politeness scale~\citep{danescu2013politeness}. We regress the politeness rating on an indicator for a
hedging expression and target its coefficient.

We use GPT-5.4 nano to construct a predictive distribution over the 25 possible ratings. For each of PPI++ and RePPI, we also evaluate two point predictions,
one using the predictive mean $\widehat Y=\E_Q(Y)$
and the other using the modal rating
$\widehat Y=\arg\max_{y\in\{1,\ldots,25\}}Q(y)$, whereas DiPPI uses the full probability vector. This
setting therefore examines whether features of the rating distribution beyond
its mean or mode contain additional information about the OLS score. The analysis population contains $4386$ requests. We set $r\in\{0.10,0.15,0.20,0.25,0.30\}$. Figure~\ref{fig:politeness} shows that DiPPI gives the shortest intervals
throughout and the lowest or nearly lowest relative RMSE, while maintaining
near-nominal coverage.
\section{Discussion}
\label{sec:discussion}
We introduce distribution-informed prediction-powered inference (DiPPI), which uses predictive distributions as auxiliary information for inference based on general estimating equations. We characterize when a predictive distribution provides information about the estimating score beyond a point prediction and quantify the resulting oracle efficiency gain.
We establish efficiency results, asymptotic guarantees for the
cross-fitted estimator, and valid inference under matrix calibration. Simulations
and real-data experiments show that DiPPI improves efficiency over PPI methods
based on point predictions when the predictive distribution contains additional
score-relevant information. %

\clearpage
\appendix
\begin{center}
{\LARGE\bfseries Appendix}
\end{center}
\section{Additional Simulation Studies}
\label{app:additional-simulation}

\subsection{Predictive variance across labeled ratios}
\label{sec:sim-spread-ratio}

We further examine the predictive-variance setting in
Section~\ref{sec:sim} as the amount of labeled data varies. We fix
$\kappa=1.5$, set $N=10{,}000$, and vary
$
r=n/N\in\{0.05,0.10,0.15,0.20,0.25,0.30\}.$
Each configuration uses $200$ replications.

\begin{figure}[ht!]
	\centering
	\includegraphics[width=0.98\linewidth]{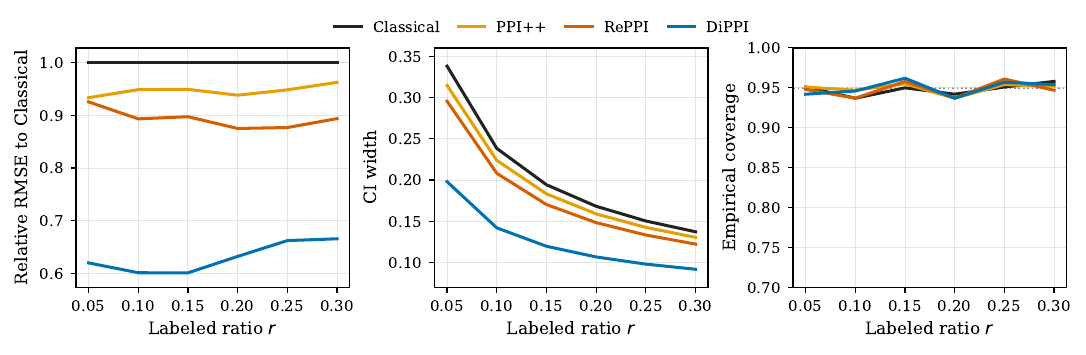}
	\caption{Outcome-relevant predictive variance as the labeled ratio varies.
		Panels show relative RMSE to Classical, mean $95\%$ CI width, and empirical
		coverage over $200$ replications.}
	\label{fig:sim-spread-ratio}
\end{figure}

Figure~\ref{fig:sim-spread-ratio} shows that DiPPI has the lowest relative
RMSE and shortest intervals across the labeled ratios, while maintaining
near-nominal coverage.

\subsection{Predictive skewness}
\label{sec:sim-skewness}

We also consider information carried by predictive skewness rather than
predictive variance. We use the methods, representation, and evaluation metrics
from Section~\ref{sec:sim}. Let $X,Z\in\mathbb R^5$ have independent
standard-normal coordinates, let $S$ be uniform on $\{-1,1\}$,
$G\sim\operatorname{Gamma}(4,1)$, and $\varepsilon\sim N(0,1)$, all
independent. Define
$W_{+}=(G-4)/2,$ and $
W_{-}=-W_{+}.
$
Given $(Z,S)$, the prediction system returns the distribution $Q$ of
$Z^\top\beta+W_S$, and we generate
\[
Y=X^\top\theta+Z^\top\beta+\kappa S+\varepsilon.
\]
Every $Q$ has the same mean $Z^\top\beta$ and variance one, but its skewness
reveals $S$. PPI++ and RePPI therefore use the same point prediction
$\widehat Y=Z^\top\beta$ regardless of $S$, whereas the asymmetric quantiles
in $\Phi(Q)$ retain the tail direction. We estimate the five OLS slopes in
$\theta$. We use $N=10{,}000$ and $200$ replications. The top row of
Figure~\ref{fig:sim-skewness} varies
$r\in\{0.05,0.10,0.15,0.20,0.25,0.30\}$ at $\kappa=1$, and the bottom row
varies $\kappa\in\{0,0.25,0.5,0.75,1\}$ at $n=1{,}000$.

\begin{figure}[ht!]
	\centering
	\includegraphics[width=0.98\textwidth]{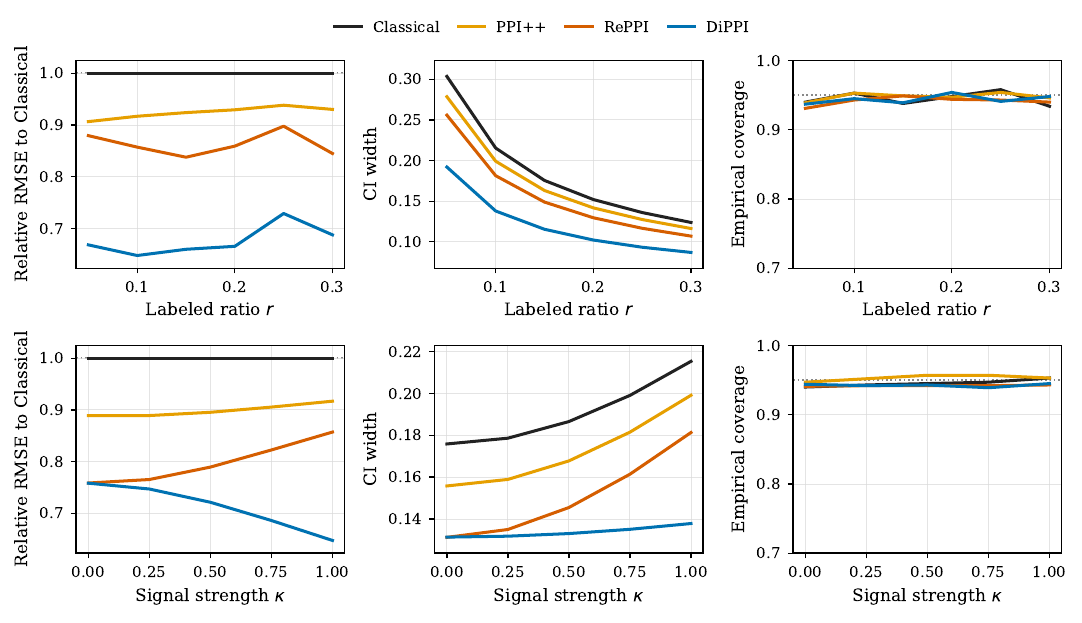}
	\caption{Outcome-relevant predictive skewness. The top row varies the labeled
		ratio at $\kappa=1$, and the bottom row varies $\kappa$ at $n=1{,}000$.
		Columns show relative RMSE to Classical, mean $95\%$ CI width, and empirical
		coverage over $200$ replications.}
	\label{fig:sim-skewness}
\end{figure}

Figure~\ref{fig:sim-skewness} shows that DiPPI and RePPI perform similarly
when $\kappa=0$. As the skewness becomes more informative about the outcome,
DiPPI achieves lower relative RMSE and shorter intervals while maintaining
near-nominal coverage.

\section{Regularity Conditions and Proofs}
\label{app:theory}

This section states sufficient conditions and proves the results in Section~\ref{sec:framework}. For each cross-fitted estimator, we condition on the data used to compute the initial estimator and estimate the conditional score. The remaining labeled fold and the unlabeled sample are then independent of the resulting estimates. Lemma~\ref{lem:foldwise-replacement} bounds the difference between their correction averages.

We use $\|\cdot\|$ for the Euclidean norm and its induced operator norm.

\subsection{Proof of Theorem~\ref{thm:fixed-correction}}

\begin{proof}
{
{We first establish existence and uniqueness of a local root,
together with its $n^{-1/2}$ rate.} Define
$$
{\Psi_{n,N}(\theta;h):={}}
\frac{1}{n}\sum_{i=1}^{n}\{U_\theta(X_i,Y_i)-h_\theta(X_i,Q_i)\}+
\frac{1}{n+N}\sum_{k=1}^{n+N}h_\theta(X_k,Q_k).
$$
\[
T_{n,N}(\theta)=\theta-H_{\theta^\star}^{-1}
\Psi_{n,N}(\theta;h).
\]
Since $H_{\theta^\star}$ is nonsingular,
\[
\Psi_{n,N}(\theta;h)=0
\quad\Longleftrightarrow\quad
T_{n,N}(\theta)=\theta.
\]
Thus, proving that the estimating equation has a unique root in
$\mathcal B_a$ is equivalent to proving that $T_{n,N}$ has a unique fixed
point in $\mathcal B_a$. Since $\mathcal B_a$ is a closed subset of
$\mathbb R^d$, it is complete. By the Banach fixed-point theorem,
$T_{n,N}$ has a unique fixed point in $\mathcal B_a$ whenever, for some
$\kappa\in(0,1)$,
\[
\begin{aligned}
T_{n,N}(\mathcal B_a)&\subseteq\mathcal B_a, \text{ and }
\|T_{n,N}(\theta_1)-T_{n,N}(\theta_2)\|
&\leq\kappa\|\theta_1-\theta_2\|,
\qquad \forall\,\theta_1,\theta_2\in\mathcal B_a.
\end{aligned}
\]
{Because $\mathcal B_a$ is convex,
$\theta_2+t(\theta_1-\theta_2)\in\mathcal B_a$ for every
$\theta_1,\theta_2\in\mathcal B_a$ and $t\in[0,1]$. The fundamental theorem
of calculus along this line segment yields}
\begin{align*}
T_{n,N}(\theta_1)-T_{n,N}(\theta_2)
={}&\left[\int_0^1\nabla_\theta T_{n,N}
\{\theta_2+t(\theta_1-\theta_2)\}\,dt\right]
(\theta_1-\theta_2).
\end{align*}
{
{
Hence,
\begin{equation}
\sup_{\theta\in\mathcal B_a}
\|\nabla_\theta T_{n,N}(\theta)\|\leq\kappa
\quad\Longrightarrow\quad
\|T_{n,N}(\theta_1)-T_{n,N}(\theta_2)\|
\leq\kappa\|\theta_1-\theta_2\|.
\label{eq:fixed-contraction-criterion}
\end{equation}
}
We now verify the derivative bound. Since
\[
\nabla_\theta T_{n,N}(\theta)
=I_d-H_{\theta^\star}^{-1}\nabla_\theta\Psi_{n,N}(\theta;h),
\]
it is enough to control $\nabla_\theta\Psi_{n,N}(\theta;h)$ uniformly. We
proceed in two steps. We first prove the uniform law
\begin{equation}
\sup_{\theta\in\mathcal B_\rho}
\left\|
\nabla_\theta\Psi_{n,N}(\theta;h)
-\E\{\nabla_\theta U_\theta(X,Y)\}
\right\|=o_p(1).
\label{eq:fixed-jacobian-ulln}
\end{equation}
We then combine \eqref{eq:fixed-jacobian-ulln} with the continuity of
$\theta\mapsto\E\{\nabla_\theta U_\theta(X,Y)\}$ at $\theta^\star$ and the
nonsingularity of $H_{\theta^\star}$ to choose $a\in(0,\rho)$ such that
\[
\Pr\!\left\{
\sup_{\theta\in\mathcal B_a}
\|\nabla_\theta T_{n,N}(\theta)\|\leq\kappa
\right\}\longrightarrow1.
\]
To prove \eqref{eq:fixed-jacobian-ulln}, write $c_n=N/(n+N)=1/(1+n/N)$. Differentiating the sample
estimating equation gives
{
\begin{align*}
\nabla_\theta\Psi_{n,N}(\theta;h)
={}&\frac1n\sum_{i=1}^n\nabla_\theta U_\theta(X_i,Y_i)
 -c_n\frac1n\sum_{i=1}^n\nabla_\theta h_\theta(X_i,Q_i)
 +c_n\frac1N\sum_{j=n+1}^{n+N}\nabla_\theta h_\theta(X_j,Q_j).
\end{align*}
}
Since the labeled and unlabeled samples have the same $(X,Q)$ marginal
distribution, the two $h$ terms have the same population expectation. Define
\begin{align*}
A_{1n}(\theta)
&=\frac1n\sum_{i=1}^n\nabla_\theta U_\theta(X_i,Y_i)
-\E\{\nabla_\theta U_\theta(X,Y)\},\\
A_{2n}(\theta)
&=\frac1n\sum_{i=1}^n\nabla_\theta h_\theta(X_i,Q_i)
-\E\{\nabla_\theta h_\theta(X,Q)\},\\
A_{3N}(\theta)
&=\frac1N\sum_{j=n+1}^{n+N}\nabla_\theta h_\theta(X_j,Q_j)
-\E\{\nabla_\theta h_\theta(X,Q)\}.
\end{align*}
{
Then we have
\[
\nabla_\theta\Psi_{n,N}(\theta;h)
-\E\{\nabla_\theta U_\theta(X,Y)\}
=A_{1n}(\theta)-c_nA_{2n}(\theta)+c_nA_{3N}(\theta),
\]
and, since $0<c_n\leq1$, the triangle inequality gives
}
\begin{align}
&\sup_{\theta\in\mathcal B_\rho}
\left\|\nabla_\theta\Psi_{n,N}(\theta;h)
-\E\{\nabla_\theta U_\theta(X,Y)\}\right\|
\leq
\sup_{\theta\in\mathcal B_\rho}\|A_{1n}(\theta)\|
+\sup_{\theta\in\mathcal B_\rho}\|A_{2n}(\theta)\|
+\sup_{\theta\in\mathcal B_\rho}\|A_{3N}(\theta)\|.
\label{eq:fixed-jacobian-decomposition}
\end{align}
We first prove that
$\sup_{\theta\in\mathcal B_\rho}\|A_{1n}(\theta)\|=o_p(1)$.
For any $\delta>0$, let
$\widetilde\theta_1,\ldots,\widetilde\theta_{K_\delta}$ be a finite
$\delta$-net of
$\mathcal B_\rho$. The triangle inequality gives
{
\begin{align}
\sup_{\theta\in\mathcal B_\rho}\|A_{1n}(\theta)\|&\leq
\max_{1\leq k\leq K_\delta}
\|A_{1n}(\widetilde\theta_k)\|+\frac1n\sum_{i=1}^n
\sup_{\substack{\theta_1,\theta_2\in\mathcal B_\rho\\
\|\theta_1-\theta_2\|\leq\delta}}
\|\nabla_{\theta_1}U_{\theta_1}(X_i,Y_i)
-\nabla_{\theta_2}U_{\theta_2}(X_i,Y_i)\|\notag\\
&\qquad+\E\!\left[
\sup_{\substack{\theta_1,\theta_2\in\mathcal B_\rho\\
\|\theta_1-\theta_2\|\leq\delta}}
\|\nabla_{\theta_1}U_{\theta_1}(X,Y)
-\nabla_{\theta_2}U_{\theta_2}(X,Y)\|
\right].
\label{eq:fixed-finite-net-bound}
\end{align}
}
For fixed $\delta$ and each $k=1,\ldots,K_\delta$,
Assumption~\ref{assump:basic} gives
\[
\E\!\left[
\|\nabla_{\widetilde\theta_k}
U_{\widetilde\theta_k}(X,Y)\|\right]
\leq
\E\!\left[
\sup_{\theta\in\mathcal B_\rho}
\|\nabla_\theta U_\theta(X,Y)\|\right]<\infty.
\]
Hence, the weak law of large numbers, applied componentwise, gives
\[
\|A_{1n}(\widetilde\theta_k)\|
=\left\|
\frac1n\sum_{i=1}^n
\nabla_{\widetilde\theta_k}U_{\widetilde\theta_k}(X_i,Y_i)
-\E\{\nabla_{\widetilde\theta_k}
U_{\widetilde\theta_k}(X,Y)\}
\right\|\xrightarrow{p}0.
\]
Since $K_\delta<\infty$, the union bound gives, for every $\varepsilon>0$,
\begin{align*}
&\Pr\!\left\{
\max_{1\leq k\leq K_\delta}
\|A_{1n}(\widetilde\theta_k)\|>\varepsilon\right\}\leq
\sum_{k=1}^{K_\delta}
\Pr\!\left\{
\|A_{1n}(\widetilde\theta_k)\|>\varepsilon\right\}
\longrightarrow0.
\end{align*}
Therefore, the maximum over the finite net is $o_p(1)$. For $\delta>0$,
define
\[
W_\delta(x,y)=
\sup_{\substack{\theta_1,\theta_2\in\mathcal B_\rho\\
\|\theta_1-\theta_2\|\leq\delta}}
\|\nabla_{\theta_1}U_{\theta_1}(x,y)
-\nabla_{\theta_2}U_{\theta_2}(x,y)\|.
\]
In~\eqref{eq:fixed-finite-net-bound}, the last two terms can be written as $
\frac1n\sum_{i=1}^nW_\delta(X_i,Y_i)
\text{ and }
\E\{W_\delta(X,Y)\}.
$
Since
\[
\E\{W_\delta(X,Y)\}
\leq2\E\!\left[
\sup_{\theta\in\mathcal B_\rho}
\|\nabla_\theta U_\theta(X,Y)\|
\right]<\infty,
\]
the weak law of large numbers gives, for every fixed $\delta>0$,
\[
\frac1n\sum_{i=1}^nW_\delta(X_i,Y_i)
-\E\{W_\delta(X,Y)\}\xrightarrow{p}0.
\]
Equivalently,
\[
\frac1n\sum_{i=1}^nW_\delta(X_i,Y_i)
=\E\{W_\delta(X,Y)\}+o_p(1).
\]
{
Substituting this result into \eqref{eq:fixed-finite-net-bound} gives
\begin{equation}
\sup_{\theta\in\mathcal B_\rho}\|A_{1n}(\theta)\|
\leq2\E\{W_\delta(X,Y)\}+o_p(1).
\label{eq:fixed-A1-bound}
\end{equation}
}
{
Almost-sure continuity on the compact set $\mathcal B_\rho$ implies
$W_\delta(X,Y)\to0$ almost surely as $\delta\to0$. Moreover,
\[
0\leq W_\delta(X,Y)
\leq2\sup_{\theta\in\mathcal B_\rho}
\|\nabla_\theta U_\theta(X,Y)\|,
\]
and the dominating random variable is integrable by
Assumption~\ref{assump:basic}. The dominated convergence theorem therefore
implies
\begin{equation}
\lim_{\delta\to0}\E\{W_\delta(X,Y)\}=0.
\label{eq:fixed-U-equicontinuity}
\end{equation}
}
{
For any $\varepsilon>0$, \eqref{eq:fixed-U-equicontinuity} permits a fixed
$\delta_\varepsilon>0$ such that
$2\E\{W_{\delta_\varepsilon}(X,Y)\}<\varepsilon/2$. By
\eqref{eq:fixed-A1-bound},
\[
\sup_{\theta\in\mathcal B_\rho}\|A_{1n}(\theta)\|
\leq \frac{\varepsilon}{2}+o_p(1),
\]
and hence
\[
\Pr\!\left\{
\sup_{\theta\in\mathcal B_\rho}\|A_{1n}(\theta)\|>\varepsilon
\right\}
\longrightarrow0.
\]
Thus,
\begin{equation}
\sup_{\theta\in\mathcal B_\rho}\|A_{1n}(\theta)\|=o_p(1).
\label{eq:fixed-U-ulln}
\end{equation}
}
Since $n/N\to r\in(0,\infty)$, both $n$ and $N$ diverge.
{By the same finite-net argument and
Assumption~\ref{assump:fixed-correction}, with
$\nabla_\theta U_\theta(X,Y)$ replaced by
$\nabla_\theta h_\theta(X,Q)$, we obtain}
{
\begin{equation}
\sup_{\theta\in\mathcal B_\rho}\|A_{2n}(\theta)\|
+\sup_{\theta\in\mathcal B_\rho}\|A_{3N}(\theta)\|=o_p(1).
\label{eq:fixed-h-ulln}
\end{equation}
Substituting \eqref{eq:fixed-U-ulln} and \eqref{eq:fixed-h-ulln} into
\eqref{eq:fixed-jacobian-decomposition} proves
\eqref{eq:fixed-jacobian-ulln}.
}

We next prove that
$\theta\mapsto\E\{\nabla_\theta U_\theta(X,Y)\}$ is continuous at
$\theta^\star$. To this end, let
$\{\theta_m\}_{m\geq1}\subset\mathcal B_\rho$ satisfy
$\theta_m\to\theta^\star$. Almost-sure continuity gives
\[
\nabla_{\theta_m}U_{\theta_m}(X,Y)
\longrightarrow
\nabla_{\theta^\star}U_{\theta^\star}(X,Y)
\quad\text{a.s.}
\]
Moreover, we have
\begin{align*}
&\|\nabla_{\theta_m}U_{\theta_m}(X,Y)
-\nabla_{\theta^\star}U_{\theta^\star}(X,Y)\|\leq
2\sup_{\theta\in\mathcal B_\rho}
\|\nabla_\theta U_\theta(X,Y)\|,
\end{align*}
whose expectation is finite by Assumption~\ref{assump:basic}. Hence, the
dominated convergence theorem gives
\[
\E\!\left[
\|\nabla_{\theta_m}U_{\theta_m}(X,Y)
-\nabla_{\theta^\star}U_{\theta^\star}(X,Y)\|
\right]\longrightarrow0.
\]
Therefore,
\begin{align*}
&\left\|
\E\{\nabla_{\theta_m}U_{\theta_m}(X,Y)\}
-\E\{\nabla_{\theta^\star}U_{\theta^\star}(X,Y)\}
\right\|\leq
\E\!\left[
\|\nabla_{\theta_m}U_{\theta_m}(X,Y)
-\nabla_{\theta^\star}U_{\theta^\star}(X,Y)\|
\right]\longrightarrow0.
\end{align*}
Since the sequence $\{\theta_m\}_{m\geq1}$ was arbitrary, the sequential
characterization of continuity implies that
$\theta\mapsto\E\{\nabla_\theta U_\theta(X,Y)\}$ is continuous at
$\theta^\star$. Moreover,
\[
I_d-H_{\theta^\star}^{-1}
\E\{\nabla_\theta U_{\theta^\star}(X,Y)\}
=I_d-H_{\theta^\star}^{-1}H_{\theta^\star}=0.
\]
Fix $\kappa\in(0,1)$. By continuity, there exists $\delta>0$ such that
\[
\left\|
\E\{\nabla_\theta U_\theta(X,Y)\}
-\E\{\nabla_\theta U_{\theta^\star}(X,Y)\}
\right\|
<\frac{\kappa}{2\|H_{\theta^\star}^{-1}\|}
\]
whenever $\|\theta-\theta^\star\|<\delta$. Choose
$a\in(0,\min\{\rho,\delta\})$. Then, for every $\theta\in\mathcal B_a$,
\begin{align*}
\left\|I_d-H_{\theta^\star}^{-1}
\E\{\nabla_\theta U_\theta(X,Y)\}\right\|&=
\left\|H_{\theta^\star}^{-1}
\left[
\E\{\nabla_\theta U_{\theta^\star}(X,Y)\}
-\E\{\nabla_\theta U_\theta(X,Y)\}
\right]\right\|\\
&\quad\leq
\|H_{\theta^\star}^{-1}\|
\left\|
\E\{\nabla_\theta U_{\theta^\star}(X,Y)\}
-\E\{\nabla_\theta U_\theta(X,Y)\}
\right\|
\leq\frac{\kappa}{2}.
\end{align*}
Hence,
\begin{equation}
\sup_{\theta\in\mathcal B_a}
\left\|I_d-H_{\theta^\star}^{-1}
\E\{\nabla_\theta U_\theta(X,Y)\}\right\|
\leq\frac{\kappa}{2}.
\label{eq:fixed-population-contraction}
\end{equation}
}
For every $\theta\in\mathcal B_a$, adding and subtracting
$H_{\theta^\star}^{-1}\E\{\nabla_\theta U_\theta(X,Y)\}$ gives
\begin{align*}
I_d-H_{\theta^\star}^{-1}
\nabla_\theta\Psi_{n,N}(\theta;h)&=
I_d-H_{\theta^\star}^{-1}
\E\{\nabla_\theta U_\theta(X,Y)\}\\
&\quad+
H_{\theta^\star}^{-1}
\left[
\E\{\nabla_\theta U_\theta(X,Y)\}
-\nabla_\theta\Psi_{n,N}(\theta;h)
\right].
\end{align*}
Therefore, the triangle inequality and $\|AB\|\leq\|A\|\|B\|$ give
\begin{align*}
\left\|I_d-H_{\theta^\star}^{-1}
\nabla_\theta\Psi_{n,N}(\theta;h)\right\|&\leq
\left\|I_d-H_{\theta^\star}^{-1}
\E\{\nabla_\theta U_\theta(X,Y)\}\right\|\\
&\quad+
\|H_{\theta^\star}^{-1}\|
\left\|\nabla_\theta\Psi_{n,N}(\theta;h)
-\E\{\nabla_\theta U_\theta(X,Y)\}\right\|.
\end{align*}
{Taking the supremum over $\theta\in\mathcal B_a$, and using
\eqref{eq:fixed-population-contraction} and
$\mathcal B_a\subset\mathcal B_\rho$, gives}
\begin{align*}
\sup_{\theta\in\mathcal B_a}
\left\|I_d-H_{\theta^\star}^{-1}
\nabla_\theta\Psi_{n,N}(\theta;h)\right\|&\leq \frac{\kappa}{2}
+\|H_{\theta^\star}^{-1}\|
\sup_{\theta\in\mathcal B_\rho}
\left\|\nabla_\theta\Psi_{n,N}(\theta;h)
-\E\{\nabla_\theta U_\theta(X,Y)\}\right\|\\
&=\frac{\kappa}{2}+o_p(1),
\end{align*}
where the last equality follows from
\eqref{eq:fixed-jacobian-ulln} and
$\|H_{\theta^\star}^{-1}\|<\infty$.
Consequently, since
$\nabla_\theta T_{n,N}(\theta)
=I_d-H_{\theta^\star}^{-1}\nabla_\theta\Psi_{n,N}(\theta;h)$, we have
{
\begin{equation}
\Pr\!\left\{
\sup_{\theta\in\mathcal B_a}
\|\nabla_\theta T_{n,N}(\theta)\|\leq\kappa
\right\}
=
\Pr\!\left\{
\sup_{\theta\in\mathcal B_a}
\left\|I_d-H_{\theta^\star}^{-1}
\nabla_\theta\Psi_{n,N}(\theta;h)\right\|\leq\kappa
\right\}\longrightarrow1.
\label{eq:fixed-contraction-event}
\end{equation}
}
It remains to verify that $T_{n,N}$ maps $\mathcal B_a$ into itself. We first
control its displacement at the center:
\[
T_{n,N}(\theta^\star)-\theta^\star
=-H_{\theta^\star}^{-1}\Psi_{n,N}(\theta^\star;h).
\]
Because $\E\{U_{\theta^\star}(X,Y)\}=0$ and the two samples have the same
$(X,Q)$ marginal distribution,
{
\begin{align*}
\Psi_{n,N}(\theta^\star;h)
={}&\frac1n\sum_{i=1}^n
\Big[U_{\theta^\star}(X_i,Y_i)-c_nh_{\theta^\star}(X_i,Q_i)
-\E\{U_{\theta^\star}(X,Y)-c_nh_{\theta^\star}(X,Q)\}\Big]\\
&+c_n\frac1N\sum_{j=n+1}^{n+N}
\Big[h_{\theta^\star}(X_j,Q_j)
-\E\{h_{\theta^\star}(X,Q)\}\Big].
\end{align*}
}
By sample independence, Assumptions~\ref{assump:basic}
and~\ref{assump:fixed-correction}, and $n/N\to r$, we have
{
\begin{align*}
\E\bigl\|\sqrt n\,\Psi_{n,N}(\theta^\star;h)\bigr\|^2
={}&\E\Bigl\|U_{\theta^\star}(X,Y)-c_nh_{\theta^\star}(X,Q)
-\E\{U_{\theta^\star}(X,Y)-c_nh_{\theta^\star}(X,Q)\}\Bigr\|^2\\
&+\frac nN c_n^2\E\Bigl\|h_{\theta^\star}(X,Q)
-\E\{h_{\theta^\star}(X,Q)\}\Bigr\|^2
\\
\leq{}&
2\E\|U_{\theta^\star}(X,Y)\|^2
+\left(2+\frac nN\right)
\E\|h_{\theta^\star}(X,Q)\|^2
=O(1).
\end{align*}
}
{
Therefore, Markov's inequality gives
\begin{equation}
\sqrt n\,\Psi_{n,N}(\theta^\star;h)=O_p(1),
\qquad
H_{\theta^\star}^{-1}\Psi_{n,N}(\theta^\star;h)
=O_p(n^{-1/2})=o_p(1).
\label{eq:fixed-center-rate}
\end{equation}
}
{
Since $(1-\kappa)a>0$, \eqref{eq:fixed-center-rate} implies
\begin{equation}
\Pr\!\left\{
\|H_{\theta^\star}^{-1}\Psi_{n,N}(\theta^\star;h)\|
\leq(1-\kappa)a
\right\}\longrightarrow1.
\label{eq:fixed-center-event}
\end{equation}
}
{
Equations \eqref{eq:fixed-contraction-event} and
\eqref{eq:fixed-center-event}, together with the union bound, give
\begin{equation}
\Pr\!\left(
\left\{
\sup_{\theta\in\mathcal B_a}
\|\nabla_\theta T_{n,N}(\theta)\|\leq\kappa
\right\}
\cap
\left\{
\|H_{\theta^\star}^{-1}\Psi_{n,N}(\theta^\star;h)\|
\leq(1-\kappa)a
\right\}
\right)\longrightarrow1.
\label{eq:fixed-banach-event}
\end{equation}
}
{
On the event inside \eqref{eq:fixed-banach-event},
\eqref{eq:fixed-contraction-criterion} shows that $T_{n,N}$ is a
$\kappa$-contraction on $\mathcal B_a$. For every
$\theta\in\mathcal B_a$, the two conditions inside
\eqref{eq:fixed-banach-event} give
\[
\|T_{n,N}(\theta)-\theta^\star\|
\leq\kappa\|\theta-\theta^\star\|
+\|H_{\theta^\star}^{-1}\Psi_{n,N}(\theta^\star;h)\|
\leq\kappa a+(1-\kappa)a=a.
\]
Thus, on the event inside \eqref{eq:fixed-banach-event}, $T_{n,N}$ is a
contraction from $\mathcal B_a$ into itself. The Banach fixed-point theorem
and \eqref{eq:fixed-banach-event} therefore yield
\begin{equation}
\Pr\{\exists!\,\theta\in\mathcal B_a:
\Psi_{n,N}(\theta;h)=0\}
\longrightarrow1.
\label{eq:fixed-root-existence}
\end{equation}
}
{
{On the event inside \eqref{eq:fixed-banach-event}, the contraction inequality
and the fixed-point identity
$\widehat\theta_h=T_{n,N}(\widehat\theta_h)$ give}
\begin{align*}
\|\widehat\theta_h-\theta^\star\|
&=\|T_{n,N}(\widehat\theta_h)-\theta^\star\|\\
&\leq
\|T_{n,N}(\widehat\theta_h)-T_{n,N}(\theta^\star)\|
+\|T_{n,N}(\theta^\star)-\theta^\star\|\\
&\leq
\kappa\|\widehat\theta_h-\theta^\star\|
+\|H_{\theta^\star}^{-1}\Psi_{n,N}(\theta^\star;h)\|,
\end{align*}
where the last step uses
$T_{n,N}(\theta^\star)-\theta^\star
=-H_{\theta^\star}^{-1}\Psi_{n,N}(\theta^\star;h)$.
Therefore,
\[
(1-\kappa)\|\widehat\theta_h-\theta^\star\|
\leq
\|H_{\theta^\star}^{-1}\Psi_{n,N}(\theta^\star;h)\|,
\]
and hence
\begin{equation}
\|\widehat\theta_h-\theta^\star\|
\leq\frac{1}{1-\kappa}
\|H_{\theta^\star}^{-1}\Psi_{n,N}(\theta^\star;h)\|
\leq\frac{\|H_{\theta^\star}^{-1}\|}{1-\kappa}
\|\Psi_{n,N}(\theta^\star;h)\|.
\label{eq:fixed-root-bound}
\end{equation}
{Equations \eqref{eq:fixed-banach-event}, \eqref{eq:fixed-center-rate}, and
\eqref{eq:fixed-root-bound} imply}
\begin{equation}
\|\widehat\theta_h-\theta^\star\|=O_p(n^{-1/2}).
\label{eq:fixed-root-rate}
\end{equation}
}
{Equation \eqref{eq:fixed-root-rate} also implies}
\[
\Pr\{\widehat\theta_h\in\operatorname{int}(\mathcal B_a)\}
=\Pr\{\|\widehat\theta_h-\theta^\star\|<a\}
\longrightarrow1.
\]

{We next derive the asymptotic distribution.} Define
\[
\overline H_{n,N}
=\int_0^1\nabla_\theta\Psi_{n,N}
\{\theta^\star+t(\widehat\theta_h-\theta^\star);h\}\,dt.
\]
{
On the event $\{\widehat\theta_h\in\mathcal B_a\}$,
\begin{align*}
\|\overline H_{n,N}-H_{\theta^\star}\|
&\leq
\sup_{\theta\in\mathcal B_\rho}
\left\|\nabla_\theta\Psi_{n,N}(\theta;h)
-\E\{\nabla_\theta U_\theta(X,Y)\}\right\|\\
&\quad+
\sup_{0\leq t\leq1}
\left\|
\E\!\left[
\nabla_{\theta^\star+t(\widehat\theta_h-\theta^\star)}
U_{\theta^\star+t(\widehat\theta_h-\theta^\star)}(X,Y)
\right]
-H_{\theta^\star}
\right\|=o_p(1),
\end{align*}
{The $o_p(1)$ comes from Equations~\eqref{eq:fixed-jacobian-ulln} and \eqref{eq:fixed-root-rate}. Specifically,
\eqref{eq:fixed-jacobian-ulln} gives
\[
\sup_{\theta\in\mathcal B_\rho}
\left\|\nabla_\theta\Psi_{n,N}(\theta;h)
-\E\{\nabla_\theta U_\theta(X,Y)\}\right\|=o_p(1).
\]
\eqref{eq:fixed-root-rate} and continuity of
$\theta\mapsto\E\{\nabla_\theta U_\theta(X,Y)\}$ give
\[
\sup_{0\leq t\leq1}
\left\|
\E\!\left[
\nabla_{\theta^\star+t(\widehat\theta_h-\theta^\star)}
U_{\theta^\star+t(\widehat\theta_h-\theta^\star)}(X,Y)
\right]-H_{\theta^\star}
\right\|=o_p(1).
\]
Hence, $\overline H_{n,N}\xrightarrow{p}H_{\theta^\star}$.
} Since
$H_{\theta^\star}$ is nonsingular, the continuous mapping theorem gives
}
{
{
\begin{equation}
\Pr\{\overline H_{n,N}\text{ is nonsingular}\}\longrightarrow1,
\qquad
\overline H_{n,N}^{-1}\xrightarrow{p}H_{\theta^\star}^{-1}.
\label{eq:fixed-average-jacobian}
\end{equation}
}
{
Equations \eqref{eq:fixed-root-existence} and
\eqref{eq:fixed-average-jacobian} imply
\begin{equation}
\Pr\!\left(
\{\Psi_{n,N}(\widehat\theta_h;h)=0\}
\cap
\{\overline H_{n,N}\text{ is nonsingular}\}
\right)\longrightarrow1.
\label{eq:fixed-linearization-event}
\end{equation}
}
{On the event inside \eqref{eq:fixed-linearization-event}, the fundamental
theorem of calculus along the line
segment from $\theta^\star$ to $\widehat\theta_h$ gives}
\[
0=\Psi_{n,N}(\theta^\star;h)
+\overline H_{n,N}(\widehat\theta_h-\theta^\star).
\]
{
Therefore,
\begin{equation}
\sqrt n(\widehat\theta_h-\theta^\star)
=-\overline H_{n,N}^{-1}\sqrt n\,
\Psi_{n,N}(\theta^\star;h).
\label{eq:fixed-exact-linearization}
\end{equation}
}
{
Equations \eqref{eq:fixed-average-jacobian} and
\eqref{eq:fixed-center-rate} give
\begin{equation}
\left\|
(\overline H_{n,N}^{-1}-H_{\theta^\star}^{-1})
\sqrt n\,\Psi_{n,N}(\theta^\star;h)
\right\|=o_p(1).
\label{eq:fixed-jacobian-replacement}
\end{equation}
}
{Combining \eqref{eq:fixed-linearization-event},
\eqref{eq:fixed-exact-linearization}, and
\eqref{eq:fixed-jacobian-replacement} yields}
\begin{equation}
\sqrt n(\widehat\theta_h-\theta^\star)
=-H_{\theta^\star}^{-1}\sqrt n\,
\Psi_{n,N}(\theta^\star;h)+o_p(1).
\label{eq:fixed-proof-linearization}
\end{equation}
}
Moreover,
{
\begin{align}
\sqrt n\,\Psi_{n,N}(\theta^\star;h)
={}&\frac1{\sqrt n}\sum_{i=1}^n
\Bigl[
U_{\theta^\star}(X_i,Y_i)-c_nh_{\theta^\star}(X_i,Q_i)
-\E\{U_{\theta^\star}(X,Y)-c_nh_{\theta^\star}(X,Q)\}
\Bigr]\notag\\
&+c_n\frac{\sqrt n}{N}\sum_{j=n+1}^{n+N}
\Bigl[
h_{\theta^\star}(X_j,Q_j)-\E\{h_{\theta^\star}(X,Q)\}
\Bigr].
\label{eq:fixed-score-decomposition}
\end{align}
}
{
Set $c=(1+r)^{-1}$. Since $n/N\to r$, we have $c_n\to c$.
By sample independence and the multivariate central limit theorem,
\eqref{eq:fixed-score-decomposition} gives
\begin{equation}
\sqrt n\,\Psi_{n,N}(\theta^\star;h)
\xrightarrow{d}
N\!\left(0,
\frac{1}{1+r}\Cov\{U_{\theta^\star}(X,Y)-h_{\theta^\star}(X,Q)\}
+\frac{r}{1+r}\Cov\{U_{\theta^\star}(X,Y)\}
\right){.}
\label{eq:fixed-score-clt}
\end{equation}
}
The covariance in~\eqref{eq:fixed-score-clt} follows from
\[
\begin{aligned}
\Cov\{U_{\theta^\star}(X,Y)-c h_{\theta^\star}(X,Q)\}
&+rc^2\Cov\{h_{\theta^\star}(X,Q)\}\\
&=c\Cov\{U_{\theta^\star}(X,Y)-h_{\theta^\star}(X,Q)\}
+(1-c)\Cov\{U_{\theta^\star}(X,Y)\}.
\end{aligned}
\]
{Equations \eqref{eq:fixed-proof-linearization} and
\eqref{eq:fixed-score-clt}, together with Slutsky's theorem,
yield}
\[
\sqrt n(\widehat\theta_h-\theta^\star)
\xrightarrow{d}
N\!\left(0,
H_{\theta^\star}^{-1}\left[
\begin{aligned}
&\frac{1}{1+r}\Cov\{U_{\theta^\star}(X,Y)-h_{\theta^\star}(X,Q)\}\\
&\quad+\frac{r}{1+r}\Cov\{U_{\theta^\star}(X,Y)\}
\end{aligned}
\right]H_{\theta^\star}^{-\top}
\right),
\]
which {proves}~\eqref{eq:general-clt}.
{It remains to prove the oracle optimality statement.}
\begin{samepage}
{For any $h$ satisfying
Assumption~\ref{assump:fixed-correction}},
{
\begin{align}
&\frac{1}{1+r}\Cov\{U_{\theta^\star}(X,Y)-h_{\theta^\star}(X,Q)\}
+\frac{r}{1+r}\Cov\{U_{\theta^\star}(X,Y)\}\notag\\
&\quad=\frac{1}{1+r}\E\!\left[
\Cov\{U_{\theta^\star}(X,Y)\mid X,Q\}\right]
+\frac{r}{1+r}\Cov\{U_{\theta^\star}(X,Y)\}\notag\\
&\qquad+\frac{1}{1+r}\Cov\!\left[
\E\{U_{\theta^\star}(X,Y)\mid X,Q\}-h_{\theta^\star}(X,Q)
\right],
\label{eq:fixed-oracle-decomposition}
\end{align}
}
\end{samepage}
\begin{samepage}
{
The last term on the right-hand side of
\eqref{eq:fixed-oracle-decomposition} is positive semidefinite and equals
zero for~\eqref{eq:oracle-h-Q}. Substituting \eqref{eq:oracle-h-Q} into
\eqref{eq:fixed-oracle-decomposition} and using the law of total covariance for
$U_{\theta^\star}(X,Y)$, yields~\eqref{eq:oracle-covariance}. More generally,
\[
\Sigma_h-\Sigma_{h^\star}
=\frac1{1+r}H_{\theta^\star}^{-1}\Cov\!\left[
\E\{U_{\theta^\star}(X,Y)\mid X,Q\}-h_{\theta^\star}(X,Q)
\right]H_{\theta^\star}^{-\top}\succeq0.
\]
This proves $\Sigma_{h^\star}\preceq\Sigma_h$.
}
\end{samepage}
}
\end{proof}

\subsection{Proof of Theorem~\ref{thm:feasible-dippi}}

{The following lemma isolates the effect of nuisance
estimation. It shows that each learned fold correction is first-order
equivalent to the correction formed with the fixed population quantities
$M_\mu$ and $\mu$. The theorem proof then treats this population correction
as the leading term.}

\begin{samepage}
\begin{lemma}
\label{lem:foldwise-replacement}
{Suppose the conditions of
Theorem~\ref{thm:feasible-dippi} hold.}
Define
\[
M_\mu=
\Cov\{U_{\theta^\star}(X,Y),\mu(X,Q)\}
\Cov\{\mu(X,Q)\}^{-1}.
\]
Then, for every $k=1,2,3$, $\widehat M^k\xrightarrow{p}M_\mu$, and
\begin{align*}
R_{n,k}:=
\sqrt n\Bigg[&
\frac{\widehat M^k}{1+n/N}\left\{
\frac1N\sum_{j=n+1}^{n+N}\widehat\mu^k(X_j,Q_j)
-\frac1{n_k}\sum_{i\in\mathcal D_k}\widehat\mu^k(X_i,Q_i)
\right\}\\
&-\frac{M_\mu}{1+r}\left\{
\frac1N\sum_{j=n+1}^{n+N}\mu(X_j,Q_j)
-\frac1{n_k}\sum_{i\in\mathcal D_k}\mu(X_i,Q_i)
\right\}\Bigg]=o_p(1).
\end{align*}
\end{lemma}
The proof of Lemma~\ref{lem:foldwise-replacement} is given in Appendix~\ref{app:proof-foldwise-replacement}.
\end{samepage}

\begin{proof}
{
The proof has three steps. We first establish existence, uniqueness, and the
$n^{-1/2}$ rate of each fold-specific root. We then derive an asymptotic
linear representation of the weighted estimator. Finally, we apply the
multivariate central limit theorem and calculate the limiting covariance.

For $k=1,2,3$, define the fold estimating function
\begin{align*}
\Psi_k(\theta)
={}&\frac1{n_k}\sum_{i\in\mathcal D_k}U_\theta(X_i,Y_i)+\frac{\widehat M^k}{1+n/N}
\left\{
\frac1N\sum_{j=n+1}^{n+N}\widehat\mu^k(X_j,Q_j)
-\frac1{n_k}\sum_{i\in\mathcal D_k}\widehat\mu^k(X_i,Q_i)
\right\}.
\end{align*}
Following the center-control step in the proof of
Theorem~\ref{thm:fixed-correction}, with $n$ replaced by $n_k$, we first
control $\Psi_k(\theta^\star)$. Since
$\E\{U_{\theta^\star}(X,Y)\}=0$, Assumption~\ref{assump:basic}, and
$n_k/n\to\pi_k\in(0,1)$ give
\[
\E\left\|
\frac{\sqrt n}{n_k}\sum_{i\in\mathcal D_k}
U_{\theta^\star}(X_i,Y_i)
\right\|^2
=\frac n{n_k}\E\|U_{\theta^\star}(X,Y)\|^2
=O(1).
\]
Markov's inequality therefore gives
{
\begin{equation}
\frac1{n_k}\sum_{i\in\mathcal D_k}
U_{\theta^\star}(X_i,Y_i)
=O_p(n^{-1/2}).
\label{eq:fold-U-center-rate}
\end{equation}
}
Since $\mu(X,Q)$ is square integrable, sample independence,
$n/N\to r$, $n_k/n\to\pi_k$, and the multivariate central limit theorem
give
\[
\sqrt n\left\{
\frac1N\sum_{j=n+1}^{n+N}\mu(X_j,Q_j)
-\frac1{n_k}\sum_{i\in\mathcal D_k}\mu(X_i,Q_i)
\right\}=O_p(1).
\]
Lemma~\ref{lem:foldwise-replacement} then gives
{
\begin{equation}
\frac{\widehat M^k}{1+n/N}
\left\{
\frac1N\sum_{j=n+1}^{n+N}\widehat\mu^k(X_j,Q_j)
-\frac1{n_k}\sum_{i\in\mathcal D_k}\widehat\mu^k(X_i,Q_i)
\right\}=O_p(n^{-1/2}).
\label{eq:fold-correction-center-rate}
\end{equation}
}
{By the definition of $\Psi_k$,
\eqref{eq:fold-U-center-rate} and
\eqref{eq:fold-correction-center-rate} imply}
\begin{equation}
\max_{1\leq k\leq3}\|\Psi_k(\theta^\star)\|
=O_p(n^{-1/2}).
\label{eq:fold-center-rate}
\end{equation}

We next verify the contraction conditions. The correction in $\Psi_k$ does
not depend on $\theta$, so
\[
\nabla_\theta\Psi_k(\theta)
=\frac1{n_k}\sum_{i\in\mathcal D_k}
\nabla_\theta U_\theta(X_i,Y_i).
\]
Since $n_k\to\infty$ and there are only three folds, the finite-net
argument used to prove \eqref{eq:fixed-jacobian-ulln}, with $n$ replaced by
$n_k$, gives
\begin{equation}
\max_{1\leq k\leq3}
\sup_{\theta\in\mathcal B_\rho}
\left\|
\nabla_\theta\Psi_k(\theta)
-\E\{\nabla_\theta U_\theta(X,Y)\}
\right\|=o_p(1).
\label{eq:fold-jacobian-ulln}
\end{equation}
Define
\[
T_k(\theta)=\theta-H_{\theta^\star}^{-1}\Psi_k(\theta).
\]
{Then choose $a\in(0,\rho)$ and $\kappa\in(0,1)$ so that the population
derivative bound in~\eqref{eq:fixed-population-contraction} holds.
Equation~\eqref{eq:fold-jacobian-ulln} then gives
\[
\Pr\!\left\{
\max_{1\leq k\leq3}\sup_{\theta\in\mathcal B_a}
\|\nabla_\theta T_k(\theta)\|\leq\kappa
\right\}\longrightarrow1.
\]
Whenever this derivative bound holds, for every
$\theta\in\mathcal B_a$ and $k=1,2,3$,
\[
\|T_k(\theta)-\theta^\star\|
\leq\kappa\|\theta-\theta^\star\|
+\|H_{\theta^\star}^{-1}\Psi_k(\theta^\star)\|.
\]
Equation~\eqref{eq:fold-center-rate} therefore gives
\[
\Pr\!\left\{
\max_{1\leq k\leq3}
\|H_{\theta^\star}^{-1}\Psi_k(\theta^\star)\|
\leq(1-\kappa)a
\right\}\longrightarrow1.
\]
Consequently,
}
\[
\Pr\!\left(
\left[\bigcap_{k=1}^3
\{T_k(\mathcal B_a)\subseteq\mathcal B_a\}\right]
\cap
\left\{
\max_{1\leq k\leq3}
\sup_{\theta\in\mathcal B_a}
\|\nabla_\theta T_k(\theta)\|\leq\kappa
\right\}
\right)\longrightarrow1.
\]
On this event, the Banach fixed-point theorem gives a unique root
$\widehat\theta^k\in\mathcal B_a$ of $\Psi_k(\theta)=0$ for every $k$.
Moreover, the contraction inequality and
$\widehat\theta^k=T_k(\widehat\theta^k)$ give, for each $k$,
\[
(1-\kappa)\|\widehat\theta^k-\theta^\star\|
\leq\|H_{\theta^\star}^{-1}\Psi_k(\theta^\star)\|.
\]
Taking the maximum over the three folds gives
\[
(1-\kappa)
\max_{1\leq k\leq3}\|\widehat\theta^k-\theta^\star\|
\leq
\max_{1\leq k\leq3}
\|H_{\theta^\star}^{-1}\Psi_k(\theta^\star)\|
\leq
\|H_{\theta^\star}^{-1}\|
\max_{1\leq k\leq3}\|\Psi_k(\theta^\star)\|.
\]
Therefore, \eqref{eq:fold-center-rate} yields
\begin{equation}
\max_{1\leq k\leq3}
\|\widehat\theta^k-\theta^\star\|=O_p(n^{-1/2}).
\label{eq:fold-root-rate}
\end{equation}
Since $\widehat\theta^{\mathrm{DiPPI}}=\sum_{k=1}^3
\omega_k\widehat\theta^k$ and $\sum_{k=1}^3\omega_k=1$,
\[
\|\widehat\theta^{\mathrm{DiPPI}}-\theta^\star\|
\leq\sum_{k=1}^3\omega_k
\|\widehat\theta^k-\theta^\star\|
=O_p(n^{-1/2}).
\]
This proves consistency.

We next derive the asymptotic linear representation. For each $k$, define
\[
\overline H_k
=\frac1{n_k}\sum_{i\in\mathcal D_k}\int_0^1
\nabla_\theta U_{\theta^\star+t(\widehat\theta^k-\theta^\star)}
(X_i,Y_i)\,dt.
\]
By \eqref{eq:fold-jacobian-ulln}, \eqref{eq:fold-root-rate}, and the
continuity of $\theta\mapsto\E\{\nabla_\theta U_\theta(X,Y)\}$ at
$\theta^\star$,
\[
\max_{1\leq k\leq3}\|\overline H_k-H_{\theta^\star}\|=o_p(1).
\]
Since $H_{\theta^\star}$ is nonsingular,
\[
\Pr\{\overline H_k\text{ is nonsingular for all }k\}\longrightarrow1,
\qquad
\max_{1\leq k\leq3}
\|\overline H_k^{-1}-H_{\theta^\star}^{-1}\|=o_p(1).
\]
The fundamental theorem of calculus along the segment from $\theta^\star$
to $\widehat\theta^k$ gives
\[
0=\Psi_k(\widehat\theta^k)
=\Psi_k(\theta^\star)
+\overline H_k(\widehat\theta^k-\theta^\star).
\]
Therefore, using \eqref{eq:fold-center-rate},
\begin{align*}
\widehat\theta^k-\theta^\star
&=-\overline H_k^{-1}\Psi_k(\theta^\star)\\
&=-H_{\theta^\star}^{-1}\Psi_k(\theta^\star)
-\{\overline H_k^{-1}-H_{\theta^\star}^{-1}\}
\Psi_k(\theta^\star)\\
&=-H_{\theta^\star}^{-1}\Psi_k(\theta^\star)
+o_p(n^{-1/2}).
\end{align*}
Substituting the definition of $\Psi_k(\theta^\star)$ and applying
Lemma~\ref{lem:foldwise-replacement} give
\begin{align}
\widehat\theta^k-\theta^\star
=-H_{\theta^\star}^{-1}\Bigg[&
\frac1{n_k}\sum_{i\in\mathcal D_k}
U_{\theta^\star}(X_i,Y_i)\notag\\
&+\frac{M_\mu}{1+r}
\left\{
\frac1N\sum_{j=n+1}^{n+N}\mu(X_j,Q_j)
-\frac1{n_k}\sum_{i\in\mathcal D_k}\mu(X_i,Q_i)
\right\}
\Bigg]+o_p(n^{-1/2}).
\label{eq:fold-linear-proof}
\end{align}
The remainder is $o_p(n^{-1/2})$ simultaneously for all three folds.

Since the folds partition the labeled sample,
\[
\sum_{k=1}^3\omega_k
\frac1{n_k}\sum_{i\in\mathcal D_k}f_i
=\frac1n\sum_{i=1}^nf_i,
\qquad
\sum_{k=1}^3\omega_k=1.
\]
Multiplying \eqref{eq:fold-linear-proof} by $\omega_k$, summing over $k$,
and using $\E\{U_{\theta^\star}(X,Y)\}=0$ yield
\begin{equation}
\begin{split}
\sqrt n(\widehat\theta^{\mathrm{DiPPI}}-\theta^\star)
=-H_{\theta^\star}^{-1}\Bigg[
&\frac1{\sqrt n}\sum_{i=1}^n
\left[
U_{\theta^\star}(X_i,Y_i)
-\frac1{1+r}M_\mu
\{\mu(X_i,Q_i)-\E\mu(X,Q)\}
\right]\\
&+\frac{\sqrt n}{N}\sum_{j=n+1}^{n+N}
\frac1{1+r}M_\mu
\{\mu(X_j,Q_j)-\E\mu(X,Q)\}
\Bigg]+o_p(1).
\end{split}
\label{eq:dippi-alr}
\end{equation}

It remains to derive the limiting distribution. The labeled-sample sum and
the unlabeled-sample sum in \eqref{eq:dippi-alr} are independent and have
mean zero. By square
integrability, the multivariate central limit theorem, and $n/N\to r$,
their limiting covariance is
\[
H_{\theta^\star}^{-1}
\left[
\Cov\!\left\{
U_{\theta^\star}(X,Y)-\frac1{1+r}M_\mu\mu(X,Q)
\right\}
+r\Cov\!\left\{\frac1{1+r}M_\mu\mu(X,Q)\right\}
\right]
H_{\theta^\star}^{-\top}.
\]
By the definition of $M_\mu$, we have
{
\begin{align}
&\Cov\{U_{\theta^\star}(X,Y),M_\mu\mu(X,Q)\}\notag\\
&\quad=\Cov\{M_\mu\mu(X,Q),U_{\theta^\star}(X,Y)\}\notag\\
&\quad=\Cov\{M_\mu\mu(X,Q)\}\notag\\
&\quad=\Cov\{U_{\theta^\star}(X,Y),\mu(X,Q)\}
\Cov\{\mu(X,Q)\}^{-1}
\Cov\{\mu(X,Q),U_{\theta^\star}(X,Y)\}.
\label{eq:proof-Mmu-covariance}
\end{align}
}
Expanding
$\Cov\{U_{\theta^\star}(X,Y)-(1+r)^{-1}M_\mu\mu(X,Q)\}
+r\Cov\{(1+r)^{-1}M_\mu\mu(X,Q)\}$
and using~\eqref{eq:proof-Mmu-covariance} gives
{
\begin{align}
&\Cov\!\left\{
U_{\theta^\star}(X,Y)-\frac1{1+r}M_\mu\mu(X,Q)
\right\}
+r\Cov\!\left\{\frac1{1+r}M_\mu\mu(X,Q)\right\}\notag\\
&\quad=\Cov\{U_{\theta^\star}(X,Y)\}
-\frac2{1+r}\Cov\{M_\mu\mu(X,Q)\}
+\frac{1+r}{(1+r)^2}\Cov\{M_\mu\mu(X,Q)\}\notag\\
&\quad=\Cov\{U_{\theta^\star}(X,Y)\}
-\frac1{1+r}\Cov\{M_\mu\mu(X,Q)\}.
\label{eq:proof-leading-covariance}
\end{align}
}
{Substituting \eqref{eq:proof-Mmu-covariance} into
\eqref{eq:proof-leading-covariance}, the limiting covariance is}
\begin{align*}
H_{\theta^\star}^{-1}
\Bigl[&\Cov\{U_{\theta^\star}(X,Y)\}
-\frac1{1+r}
\Cov\{U_{\theta^\star}(X,Y),\mu(X,Q)\}
\Cov\{\mu(X,Q)\}^{-1}\\
&\quad\times
\Cov\{\mu(X,Q),U_{\theta^\star}(X,Y)\}\Bigr]
H_{\theta^\star}^{-\top}
=\Sigma_\mu^{\mathrm{DiPPI}}.
\end{align*}
Slutsky's theorem therefore yields
\[
\sqrt n(\widehat\theta^{\mathrm{DiPPI}}-\theta^\star)
\xrightarrow{d}
N\!\left(0,\Sigma_\mu^{\mathrm{DiPPI}}\right),
\]
where $\Sigma_\mu^{\mathrm{DiPPI}}$ is given in
\eqref{eq:dippi-covariance}.
}
\end{proof}

\subsection{Proof of Lemma~\ref{lem:foldwise-replacement}}
\label{app:proof-foldwise-replacement}

\begin{proof}
{
Fix $k$ and condition on the two folds used to construct
$\widehat\theta_0^k$ and $\widehat\mu^k$. The fitted quantities are then
fixed, while $\mathcal D_k$ and the unlabeled sample remain independent.
Expectations involving fitted quantities below are taken over an independent
observation, conditional on the training folds.

We first prove $\widehat M^k\xrightarrow{p}M_\mu$. By
\eqref{eq:M-hat},
\[
\widehat M^k
=\widehat{\Cov}_{\mathcal D_k}
\{U_{\widehat\theta_0^k}(X,Y),\widehat\mu^k(X,Q)\}
\left[\widehat{\Cov}_{\mathcal D_k}
\{\widehat\mu^k(X,Q)\}\right]^{-1}.
\]
Since $\Cov\{\mu(X,Q)\}\succ0$, continuity of matrix inversion shows
that it suffices to prove
\begin{align}
\widehat{\Cov}_{\mathcal D_k}
\{U_{\widehat\theta_0^k}(X,Y),\widehat\mu^k(X,Q)\}
&=\Cov\{U_{\theta^\star}(X,Y),\mu(X,Q)\}+o_p(1),
\notag\\
\widehat{\Cov}_{\mathcal D_k}\{\widehat\mu^k(X,Q)\}
&=\Cov\{\mu(X,Q)\}+o_p(1).
\label{eq:proof-covariance-targets}
\end{align}
To prove~\eqref{eq:proof-covariance-targets}, write the two empirical
covariances as
\begin{equation}
\begin{aligned}
\widehat{\Cov}_{\mathcal D_k}
\{U_{\widehat\theta_0^k}(X,Y),\widehat\mu^k(X,Q)\}&=
\frac1{n_k}\sum_{i\in\mathcal D_k}
U_{\widehat\theta_0^k}(X_i,Y_i)\widehat\mu^k(X_i,Q_i)^\top\\
&-
\left\{\frac1{n_k}\sum_{i\in\mathcal D_k}
U_{\widehat\theta_0^k}(X_i,Y_i)\right\}
\left\{\frac1{n_k}\sum_{i\in\mathcal D_k}
\widehat\mu^k(X_i,Q_i)\right\}^{\!\top},\\
\widehat{\Cov}_{\mathcal D_k}\{\widehat\mu^k(X,Q)\}
&=
\frac1{n_k}\sum_{i\in\mathcal D_k}
\widehat\mu^k(X_i,Q_i)\widehat\mu^k(X_i,Q_i)^\top\\
&\qquad-
\left\{\frac1{n_k}\sum_{i\in\mathcal D_k}
\widehat\mu^k(X_i,Q_i)\right\}
\left\{\frac1{n_k}\sum_{i\in\mathcal D_k}
\widehat\mu^k(X_i,Q_i)\right\}^{\!\top}.
\end{aligned}
\label{eq:proof-empirical-covariance-expansions}
\end{equation}
Thus, to prove~\eqref{eq:proof-covariance-targets}, it remains
to establish convergence of the four empirical moments displayed in
\eqref{eq:proof-empirical-covariance-expansions}. We first derive the
empirical $L_2$ bounds used for those four limits. Since every fold size is of order $n$, applying
Theorem~\ref{thm:fixed-correction} with
$h_\theta\equiv0$ to the three pilot estimating equations gives
\[
\max_{1\leq k\leq3}
\|\widehat\theta_0^k-\theta^\star\|=O_p(n^{-1/2}).
\]
On the event $\widehat\theta_0^k\in\mathcal B_\rho$, the integral
mean-value theorem and Assumption~\ref{assump:nuisance} give
\begin{align*}
&\E\|U_{\widehat\theta_0^k}(X,Y)
-U_{\theta^\star}(X,Y)\|^2\leq
\|\widehat\theta_0^k-\theta^\star\|^2
\E\!\left\{
\sup_{\theta\in\mathcal B_\rho}
\|\nabla_\theta U_\theta(X,Y)\|^2
\right\}
=o_p(1).
\end{align*}

Cross-fitting makes $\mathcal D_k$ independent of
$(\widehat\theta_0^k,\widehat\mu^k)$. Thus, for every
$\varepsilon>0$, conditional Markov's inequality gives
\begin{align*}
\Pr\!\left\{
\frac1{n_k}\sum_{i\in\mathcal D_k}
\|\widehat\mu^k(X_i,Q_i)-\mu(X_i,Q_i)\|^2>\varepsilon
\,\middle|\,\widehat\mu^k
\right\}&\leq\frac1\varepsilon
\E\|\widehat\mu^k(X,Q)-\mu(X,Q)\|^2\\
&=o_p(1),\\
\Pr\!\left\{
\frac1{n_k}\sum_{i\in\mathcal D_k}
\|U_{\widehat\theta_0^k}(X_i,Y_i)
-U_{\theta^\star}(X_i,Y_i)\|^2>\varepsilon
\,\middle|\,\widehat\theta_0^k
\right\}&\leq\frac1\varepsilon
\E\|U_{\widehat\theta_0^k}(X,Y)
-U_{\theta^\star}(X,Y)\|^2\\
&=o_p(1).
\end{align*}
 Therefore, we can obtain
\begin{align}
&\frac1{n_k}\sum_{i\in\mathcal D_k}
\|\widehat\mu^k(X_i,Q_i)-\mu(X_i,Q_i)\|^2=o_p(1),
\notag\\
&\frac1{n_k}\sum_{i\in\mathcal D_k}
\|U_{\widehat\theta_0^k}(X_i,Y_i)
-U_{\theta^\star}(X_i,Y_i)\|^2=o_p(1).
\label{eq:proof-empirical-l2}
\end{align}

Assumption~\ref{assump:basic} gives
$\E\|U_{\theta^\star}(X,Y)\|^2<\infty$, while
Theorem~\ref{thm:feasible-dippi} assumes
$\E\|\mu(X,Q)\|^2<\infty$. Hence, the ordinary law of large numbers gives
{
\begin{align}
\frac1{n_k}\sum_{i\in\mathcal D_k}\|\mu(X_i,Q_i)\|^2&=O_p(1),&
\frac1{n_k}\sum_{i\in\mathcal D_k}
\|U_{\theta^\star}(X_i,Y_i)\|^2&=O_p(1).
\label{eq:proof-oracle-empirical-l2}
\end{align}
}
Moreover,
{
\begin{align}
\frac1{n_k}\sum_{i\in\mathcal D_k}
\|\widehat\mu^k(X_i,Q_i)\|^2
&\leq
\frac2{n_k}\sum_{i\in\mathcal D_k}
\|\widehat\mu^k(X_i,Q_i)-\mu(X_i,Q_i)\|^2\notag\\
&\quad+\frac2{n_k}\sum_{i\in\mathcal D_k}\|\mu(X_i,Q_i)\|^2
\notag\\
&=O_p(1).
\label{eq:proof-fitted-empirical-l2}
\end{align}
}
The triangle inequality and the empirical Cauchy--Schwarz inequality give
\begin{align}
&\left\|
\frac1{n_k}\sum_{i\in\mathcal D_k}
U_{\widehat\theta_0^k}(X_i,Y_i)\widehat\mu^k(X_i,Q_i)^\top
-\frac1{n_k}\sum_{i\in\mathcal D_k}
U_{\theta^\star}(X_i,Y_i)\mu(X_i,Q_i)^\top
\right\|_F\notag\\
&\quad\leq
\left\{
\frac1{n_k}\sum_{i\in\mathcal D_k}
\|U_{\widehat\theta_0^k}(X_i,Y_i)
-U_{\theta^\star}(X_i,Y_i)\|^2
\right\}^{1/2}
\left\{
\frac1{n_k}\sum_{i\in\mathcal D_k}
\|\widehat\mu^k(X_i,Q_i)\|^2
\right\}^{1/2}\notag\\
&\qquad+
\left\{
\frac1{n_k}\sum_{i\in\mathcal D_k}
\|U_{\theta^\star}(X_i,Y_i)\|^2
\right\}^{1/2}
\left\{
\frac1{n_k}\sum_{i\in\mathcal D_k}
\|\widehat\mu^k(X_i,Q_i)-\mu(X_i,Q_i)\|^2
\right\}^{1/2}\notag\\
&\quad=o_p(1)O_p(1)+O_p(1)o_p(1)=o_p(1).
\label{eq:proof-cross-moment-replacement}
\end{align}
Similarly,
\begin{align}
&\left\|
\frac1{n_k}\sum_{i\in\mathcal D_k}
\widehat\mu^k(X_i,Q_i)\widehat\mu^k(X_i,Q_i)^\top
-\frac1{n_k}\sum_{i\in\mathcal D_k}
\mu(X_i,Q_i)\mu(X_i,Q_i)^\top
\right\|_F\notag\\
&\quad\leq
\left\{
\frac1{n_k}\sum_{i\in\mathcal D_k}
\|\widehat\mu^k(X_i,Q_i)-\mu(X_i,Q_i)\|^2
\right\}^{1/2}\notag\\
&\qquad\times
\left[
\left\{
\frac1{n_k}\sum_{i\in\mathcal D_k}
\|\widehat\mu^k(X_i,Q_i)\|^2
\right\}^{1/2}
+
\left\{
\frac1{n_k}\sum_{i\in\mathcal D_k}
\|\mu(X_i,Q_i)\|^2
\right\}^{1/2}
\right]\notag\\
&\quad=o_p(1)\{O_p(1)+O_p(1)\}=o_p(1).
\label{eq:proof-second-moment-replacement}
\end{align}
For the empirical means,
\begin{align}
\left\|
\frac1{n_k}\sum_{i\in\mathcal D_k}
\{U_{\widehat\theta_0^k}(X_i,Y_i)
-U_{\theta^\star}(X_i,Y_i)\}
\right\|&\leq
\frac1{n_k}\sum_{i\in\mathcal D_k}
\|U_{\widehat\theta_0^k}(X_i,Y_i)
-U_{\theta^\star}(X_i,Y_i)\|\notag\\
&\leq
\left\{
\frac1{n_k}\sum_{i\in\mathcal D_k}
\|U_{\widehat\theta_0^k}(X_i,Y_i)
-U_{\theta^\star}(X_i,Y_i)\|^2
\right\}^{1/2}\notag\\
&=o_p(1),
\label{eq:proof-U-mean-replacement}\\
\left\|
\frac1{n_k}\sum_{i\in\mathcal D_k}
\{\widehat\mu^k(X_i,Q_i)-\mu(X_i,Q_i)\}
\right\|&\leq
\frac1{n_k}\sum_{i\in\mathcal D_k}
\|\widehat\mu^k(X_i,Q_i)-\mu(X_i,Q_i)\|\notag\\
&\leq
\left\{
\frac1{n_k}\sum_{i\in\mathcal D_k}
\|\widehat\mu^k(X_i,Q_i)-\mu(X_i,Q_i)\|^2
\right\}^{1/2}\notag\\
&=o_p(1).
\label{eq:proof-mu-mean-replacement}
\end{align}
{Equations~\eqref{eq:proof-cross-moment-replacement}--\eqref{eq:proof-mu-mean-replacement}
follow from \eqref{eq:proof-empirical-l2},
\eqref{eq:proof-oracle-empirical-l2}, and
\eqref{eq:proof-fitted-empirical-l2}.} We next show that the corresponding oracle empirical moments
converge to their population values.
The same moment conditions and Cauchy--Schwarz give
\begin{align*}
\E\|U_{\theta^\star}(X,Y)\mu(X,Q)^\top\|
&\leq
\{\E\|U_{\theta^\star}(X,Y)\|^2\}^{1/2}
\{\E\|\mu(X,Q)\|^2\}^{1/2}<\infty,\\
\E\|\mu(X,Q)\mu(X,Q)^\top\|
&=\E\|\mu(X,Q)\|^2<\infty.
\end{align*}
Combining the ordinary law of large numbers with
\eqref{eq:proof-cross-moment-replacement}--\eqref{eq:proof-mu-mean-replacement}
gives
\begin{equation}
\begin{aligned}
\frac1{n_k}\sum_{i\in\mathcal D_k}
U_{\widehat\theta_0^k}(X_i,Y_i)\widehat\mu^k(X_i,Q_i)^\top
&=\E\{U_{\theta^\star}(X,Y)\mu(X,Q)^\top\}+o_p(1),\\
\frac1{n_k}\sum_{i\in\mathcal D_k}
\widehat\mu^k(X_i,Q_i)\widehat\mu^k(X_i,Q_i)^\top
&=\E\{\mu(X,Q)\mu(X,Q)^\top\}+o_p(1),\\
\frac1{n_k}\sum_{i\in\mathcal D_k}
U_{\widehat\theta_0^k}(X_i,Y_i)
&=\E\{U_{\theta^\star}(X,Y)\}+o_p(1),\\
\frac1{n_k}\sum_{i\in\mathcal D_k}\widehat\mu^k(X_i,Q_i)
&=\E\{\mu(X,Q)\}+o_p(1).
\end{aligned}
\label{eq:proof-fitted-moment-limits}
\end{equation}
Substituting \eqref{eq:proof-fitted-moment-limits} into
\eqref{eq:proof-empirical-covariance-expansions} proves
\eqref{eq:proof-covariance-targets}.
Since $\Cov\{\mu(X,Q)\}\succ0$, continuity of matrix inversion gives
\[
\left[
\widehat{\Cov}_{\mathcal D_k}\{\widehat\mu^k(X,Q)\}
\right]^{-1}
=\Cov\{\mu(X,Q)\}^{-1}+o_p(1).
\]
Hence,
\begin{equation}
\widehat M^k
=M_\mu+o_p(1)\xrightarrow{p}M_\mu,
\qquad
\|\widehat M^k\|=O_p(1).
\label{eq:proof-M-consistency}
\end{equation}

It remains to replace the fitted correction by its population limit.
Conditional on the training folds, the labeled and unlabeled averages of
$\widehat\mu^k(X,Q)-\mu(X,Q)$ have the same mean. Their independence gives
\begin{align*}
&\E\Bigg\|
\sqrt n\Bigg[
\frac1N\sum_{j=n+1}^{n+N}
\{\widehat\mu^k(X_j,Q_j)-\mu(X_j,Q_j)\}
-\frac1{n_k}\sum_{i\in\mathcal D_k}
\{\widehat\mu^k(X_i,Q_i)-\mu(X_i,Q_i)\}
\Bigg]\Bigg\|^2\\
&\quad=
\left(\frac nN+\frac n{n_k}\right)
\E\Big\|
\widehat\mu^k(X,Q)-\mu(X,Q)
-\E\{\widehat\mu^k(X,Q)-\mu(X,Q)\}
\Big\|^2\\
&\quad\leq
\left(\frac nN+\frac n{n_k}\right)
\E\|\widehat\mu^k(X,Q)-\mu(X,Q)\|^2
=o_p(1).
\end{align*}
Conditional Markov's inequality and boundedness of conditional probabilities
therefore imply
\begin{equation}
\sqrt n\Bigg[
\frac1N\sum_{j=n+1}^{n+N}
\{\widehat\mu^k(X_j,Q_j)-\mu(X_j,Q_j)\}
-\frac1{n_k}\sum_{i\in\mathcal D_k}
\{\widehat\mu^k(X_i,Q_i)-\mu(X_i,Q_i)\}
\Bigg]=o_p(1).
\label{eq:proof-score-replacement}
\end{equation}

Since $\E\|\mu(X,Q)\|^2<\infty$, the multivariate central limit theorem gives
\begin{equation}
\begin{aligned}
&\sqrt n\left\{
\frac1N\sum_{j=n+1}^{n+N}\mu(X_j,Q_j)
-\frac1{n_k}\sum_{i\in\mathcal D_k}\mu(X_i,Q_i)
\right\}\\
&\quad=
\sqrt{\frac nN}\frac1{\sqrt N}
\sum_{j=n+1}^{n+N}\{\mu(X_j,Q_j)-\E\mu(X,Q)\}\\
&\qquad-
\sqrt{\frac n{n_k}}\frac1{\sqrt{n_k}}
\sum_{i\in\mathcal D_k}\{\mu(X_i,Q_i)-\E\mu(X,Q)\}
=O_p(1).
\end{aligned}
\label{eq:proof-limit-score-order}
\end{equation}

Finally, by the definition of $R_{n,k}$ in Lemma~\ref{lem:foldwise-replacement},
\begin{align*}
R_{n,k}
={}&
\frac{\widehat M^k}{1+n/N}\sqrt n\Bigg[
\frac1N\sum_{j=n+1}^{n+N}
\{\widehat\mu^k(X_j,Q_j)-\mu(X_j,Q_j)\}
-\frac1{n_k}\sum_{i\in\mathcal D_k}
\{\widehat\mu^k(X_i,Q_i)-\mu(X_i,Q_i)\}
\Bigg]\\
&\qquad+
\frac{\widehat M^k-M_\mu}{1+n/N}\sqrt n\left\{
\frac1N\sum_{j=n+1}^{n+N}\mu(X_j,Q_j)
-\frac1{n_k}\sum_{i\in\mathcal D_k}\mu(X_i,Q_i)
\right\}\\
&\qquad+
\left\{\frac1{1+n/N}-\frac1{1+r}\right\}M_\mu
\sqrt n\left\{
\frac1N\sum_{j=n+1}^{n+N}\mu(X_j,Q_j)
-\frac1{n_k}\sum_{i\in\mathcal D_k}\mu(X_i,Q_i)
\right\}.
\end{align*}
By \eqref{eq:proof-M-consistency}, \eqref{eq:proof-score-replacement}, and
\eqref{eq:proof-limit-score-order}, the three terms on the right are
$o_p(1)$. Therefore, $R_{n,k}=o_p(1)$, which proves the lemma.
}
\end{proof}

\subsection{Proof of Corollary~\ref{cor:dippi-safety}}

\begin{proof}
Because $\Cov\{\mu(X,Q)\}$ is positive definite,
\[
\begin{split}
&a^\top
\Cov\{U_{\theta^\star}(X,Y),\mu(X,Q)\}
\Cov\{\mu(X,Q)\}^{-1}
\Cov\{\mu(X,Q),U_{\theta^\star}(X,Y)\}a\\
&\qquad=
\left\|
\Cov\{\mu(X,Q)\}^{-1/2}
\Cov\{\mu(X,Q),U_{\theta^\star}(X,Y)\}a
\right\|^2\geq0
\end{split}
\]
for every $a\in\R^d$.
Subtracting~\eqref{eq:dippi-covariance} from the labeled-only covariance gives~\eqref{eq:dippi-safety-gap}. Premultiplication by $H_{\theta^\star}^{-1}$ and postmultiplication by $H_{\theta^\star}^{-\top}$ preserve positive semidefiniteness.

{
For the oracle statement, suppose
$\mu(X,Q)=\E\{U_{\theta^\star}(X,Y)\mid X,\Phi(Q)\}$. By iterated
expectations and $\E\{U_{\theta^\star}(X,Y)\}=0$,
\[
\E\{\mu(X,Q)\}=\E\{U_{\theta^\star}(X,Y)\}=0.
\]
Since $\mu(X,Q)$ is measurable with respect to $(X,\Phi(Q))$, another
application of iterated expectations gives
\begin{align*}
\Cov\{U_{\theta^\star}(X,Y),\mu(X,Q)\}
={}&\E\{U_{\theta^\star}(X,Y)\mu(X,Q)^\top\}\\
={}&\E\!\left[
\E\{U_{\theta^\star}(X,Y)\mid X,\Phi(Q)\}\mu(X,Q)^\top
\right]\\
={}&\E\{\mu(X,Q)\mu(X,Q)^\top\}\\
={}&\Cov\{\mu(X,Q)\}.
\end{align*}
Thus, $M_\mu=I_d$. Substituting this into
\eqref{eq:dippi-covariance} gives \eqref{eq:dippi-oracle-covariance}.
Theorem~\ref{thm:fixed-correction}, applied with $Q$ replaced by $\Phi(Q)$,
shows that \eqref{eq:dippi-oracle-covariance} is the minimum asymptotic covariance
among corrections measurable with respect to $(X,\Phi(Q))$
}
\end{proof}

\subsection{Proof of Theorem~\ref{thm:dippi-variance}}

\begin{proof}
{
We first prove covariance consistency in
\eqref{eq:dippi-variance-consistency}. Since
$\widehat H\xrightarrow{p}H_{\theta^\star}$ and $H_{\theta^\star}$ is
nonsingular, the continuous mapping theorem gives
\[
\widehat H^{-1}\xrightarrow{p}H_{\theta^\star}^{-1}.
\]
Thus, in view of~\eqref{eq:dippi-sandwich-estimator} and $n/N\to r$, it
suffices to prove
\begin{equation}
\begin{split}
    \sum_{k=1}^3\omega_k\widehat{\Cov}_{\mathcal D_k}
    \left\{U_{\widehat\theta^{\mathrm{DiPPI}}}(X_i,Y_i)
    -\frac{\widehat M^k\widehat\mu^k(X_i,Q_i)}{1+n/N}\right\}
    \xrightarrow{p}{}&
    \Cov\left\{U_{\theta^\star}(X,Y)
    -\frac{M_\mu\mu(X,Q)}{1+r}\right\}
\end{split}
\label{eq:variance-proof-labeled-covariance}
\end{equation}
and
\begin{equation}
\begin{split}
\widehat{\Cov}_{\mathrm U}\left\{
\frac1{1+n/N}\sum_{k=1}^3\omega_k
\widehat M^k\widehat\mu^k(X_j,Q_j)
\right\}
\xrightarrow{p}
\Cov\left\{\frac1{1+r}M_\mu\mu(X,Q)\right\}.
\end{split}
\label{eq:variance-proof-unlabeled-covariance}
\end{equation}
After proving \eqref{eq:variance-proof-labeled-covariance} and
\eqref{eq:variance-proof-unlabeled-covariance}, we combine the resulting covariance
consistency with Theorem~\ref{thm:feasible-dippi} and Slutsky's theorem to
prove \eqref{eq:dippi-studentized-clt} and the stated Wald coverage.
}

{We first prove
\eqref{eq:variance-proof-labeled-covariance}.}
After multiplication by $\omega_k=|\mathcal D_k|/n$, every observation in $\mathcal D_k$ has weight $1/n$ in~\eqref{eq:dippi-alr}. The labeled folds are disjoint, so their covariance contributions add with weights $\omega_k$. The same unlabeled observations are used for all three fold estimates, so the three unlabeled corrections are averaged before their covariance is computed, as in~\eqref{eq:dippi-sandwich-estimator}.

For each fixed $k$, write
the difference between the fitted and limiting corrections as
\begin{align*}
&\frac1{1+n/N}\widehat M^k
\{\widehat\mu^k(X_i,Q_i)-\mu(X_i,Q_i)\}+
\left\{\frac1{1+n/N}\widehat M^k-\frac1{1+r}M_\mu\right\}
\mu(X_i,Q_i).
\end{align*}
Their empirical squared norms satisfy
\begin{align*}
&\frac1{n_k}\sum_{i\in\mathcal D_k}
\left\|
\frac{\widehat M^k}{1+n/N}
\{\widehat\mu^k(X_i,Q_i)-\mu(X_i,Q_i)\}
\right\|^2\\
&\quad\leq
\frac{\|\widehat M^k\|^2}{(1+n/N)^2}
\times
\frac1{n_k}\sum_{i\in\mathcal D_k}
\|\widehat\mu^k(X_i,Q_i)-\mu(X_i,Q_i)\|^2
=o_p(1),
\end{align*}
\begin{align*}
&\frac1{n_k}\sum_{i\in\mathcal D_k}
\left\|
\left\{\frac{\widehat M^k}{1+n/N}-\frac{M_\mu}{1+r}\right\}
\mu(X_i,Q_i)
\right\|^2\\
&\quad\leq
\left\|\frac{\widehat M^k}{1+n/N}-\frac{M_\mu}{1+r}\right\|^2\times
\frac1{n_k}\sum_{i\in\mathcal D_k}\|\mu(X_i,Q_i)\|^2
=o_p(1).
\end{align*}
Equation~\eqref{eq:proof-empirical-l2} and
$\|\widehat M^k\|=O_p(1)$ prove the first bound. Equation
\eqref{eq:proof-M-consistency}, $n/N\to r$, and
\[
\frac1{n_k}\sum_{i\in\mathcal D_k}\|\mu(X_i,Q_i)\|^2=O_p(1)
\]
prove the second bound. Combining them and using the fixed number of folds
gives
\begin{equation}
\sum_{k=1}^3\omega_k\frac1{n_k}\sum_{i\in\mathcal D_k}
\left\|
\frac{\widehat M^k\widehat\mu^k(X_i,Q_i)}{1+n/N}
-\frac{M_\mu\mu(X_i,Q_i)}{1+r}
\right\|^2=o_p(1).
\label{eq:variance-proof-labeled-correction}
\end{equation}

{
Theorem~\ref{thm:feasible-dippi} gives
$\widehat\theta^{\mathrm{DiPPI}}\xrightarrow{p}\theta^\star$; hence,
\[
\Pr\{\widehat\theta^{\mathrm{DiPPI}}
\in\operatorname{int}(\mathcal B_\rho)\}
=\Pr\{\|\widehat\theta^{\mathrm{DiPPI}}-\theta^\star\|<\rho\}
\longrightarrow1.
\]
}
On the event $\widehat\theta^{\mathrm{DiPPI}}\in\mathcal B_\rho$, a first-order Taylor expansion in integral form and Assumptions~\ref{assump:basic} and~\ref{assump:nuisance} give
\begin{align}
&\frac1n\sum_{i=1}^n
\|U_{\widehat\theta^{\mathrm{DiPPI}}}(X_i,Y_i)
-U_{\theta^\star}(X_i,Y_i)\|^2\notag\\
&\quad\leq
\|\widehat\theta^{\mathrm{DiPPI}}-\theta^\star\|^2
\frac1n\sum_{i=1}^n
\sup_{\xi\in\mathcal B_\rho}
\|\nabla_\theta U_\xi(X_i,Y_i)\|^2\notag\\
&=o_p(1).
\label{eq:variance-proof-score-l2}
\end{align}
{
Since $\omega_k/n_k=1/n$,
\eqref{eq:variance-proof-score-l2} and
\eqref{eq:variance-proof-labeled-correction} imply
\begin{align}
&\sum_{k=1}^3\omega_k\frac1{n_k}
\sum_{i\in\mathcal D_k}
\Bigg\|
\begin{aligned}
&U_{\widehat\theta^{\mathrm{DiPPI}}}(X_i,Y_i)
-\frac{\widehat M^k\widehat\mu^k(X_i,Q_i)}{1+n/N}-U_{\theta^\star}(X_i,Y_i)
+\frac{M_\mu\mu(X_i,Q_i)}{1+r}
\end{aligned}
\Bigg\|^2\notag\\
&\quad\leq
\frac2n\sum_{i=1}^n
\|U_{\widehat\theta^{\mathrm{DiPPI}}}(X_i,Y_i)
-U_{\theta^\star}(X_i,Y_i)\|^2\notag\\
&\qquad+2\sum_{k=1}^3\omega_k\frac1{n_k}
\sum_{i\in\mathcal D_k}
\left\|
\frac{\widehat M^k\widehat\mu^k(X_i,Q_i)}{1+n/N}
-\frac{M_\mu\mu(X_i,Q_i)}{1+r}
\right\|^2
=o_p(1).
\label{eq:variance-proof-labeled-summand}
\end{align}
Since every term on the left-hand side of
\eqref{eq:variance-proof-labeled-summand} is nonnegative and
$\omega_k^{-1}\to\pi_k^{-1}<\infty$, for each $k=1,2,3$,
\begin{align}
&\frac1{n_k}\sum_{i\in\mathcal D_k}
\left\|
U_{\widehat\theta^{\mathrm{DiPPI}}}(X_i,Y_i)
-\frac{\widehat M^k\widehat\mu^k(X_i,Q_i)}{1+n/N}
-U_{\theta^\star}(X_i,Y_i)
+\frac{M_\mu\mu(X_i,Q_i)}{1+r}
\right\|^2=o_p(1).
\label{eq:variance-proof-foldwise-labeled-l2}
\end{align}
Moreover,
\begin{align*}
&\E\left\|
U_{\theta^\star}(X,Y)-\frac{M_\mu\mu(X,Q)}{1+r}
\right\|^2\leq
2\E\|U_{\theta^\star}(X,Y)\|^2
+\frac{2\|M_\mu\|^2}{(1+r)^2}\E\|\mu(X,Q)\|^2
<\infty.
\end{align*}
Since $n_k/n\to\pi_k>0$ implies $n_k\to\infty$, the weak law of large
numbers therefore gives
\begin{equation}
\frac1{n_k}\sum_{i\in\mathcal D_k}
\left\|
U_{\theta^\star}(X_i,Y_i)
-\frac{M_\mu\mu(X_i,Q_i)}{1+r}
\right\|^2
\xrightarrow{p}
\E\left\|
U_{\theta^\star}(X,Y)-\frac{M_\mu\mu(X,Q)}{1+r}
\right\|^2<\infty,
\label{eq:variance-proof-labeled-limit-second-moment}
\end{equation}
Thus, the empirical second moment in
\eqref{eq:variance-proof-labeled-limit-second-moment} is $O_p(1)$.
The triangle inequality and
\eqref{eq:variance-proof-foldwise-labeled-l2} give
\begin{align}
&\frac1{n_k}\sum_{i\in\mathcal D_k}
\left\|
U_{\widehat\theta^{\mathrm{DiPPI}}}(X_i,Y_i)
-\frac{\widehat M^k\widehat\mu^k(X_i,Q_i)}{1+n/N}
\right\|^2\\
&\quad\leq
\frac2{n_k}\sum_{i\in\mathcal D_k}
\left\|
U_{\widehat\theta^{\mathrm{DiPPI}}}(X_i,Y_i)
-\frac{\widehat M^k\widehat\mu^k(X_i,Q_i)}{1+n/N}
-U_{\theta^\star}(X_i,Y_i)
+\frac{M_\mu\mu(X_i,Q_i)}{1+r}
\right\|^2\\
&\qquad+
\frac2{n_k}\sum_{i\in\mathcal D_k}
\left\|
U_{\theta^\star}(X_i,Y_i)
-\frac{M_\mu\mu(X_i,Q_i)}{1+r}
\right\|^2
=O_p(1).
\label{eq:variance-proof-labeled-fitted-second-moment}
\end{align}
Applying the empirical Cauchy--Schwarz inequality separately to the
empirical means and second moments, using
\eqref{eq:variance-proof-foldwise-labeled-l2},
\eqref{eq:variance-proof-labeled-limit-second-moment}, and
\eqref{eq:variance-proof-labeled-fitted-second-moment}, gives
\begin{align}
&\Bigg\|
\widehat{\Cov}_{\mathcal D_k}\left\{
U_{\widehat\theta^{\mathrm{DiPPI}}}(X_i,Y_i)
-\frac{\widehat M^k\widehat\mu^k(X_i,Q_i)}{1+n/N}
\right\}\\
&\qquad-
\widehat{\Cov}_{\mathcal D_k}\left\{
U_{\theta^\star}(X_i,Y_i)
-\frac{M_\mu\mu(X_i,Q_i)}{1+r}
\right\}
\Bigg\|_F=o_p(1).
\label{eq:variance-proof-labeled-covariance-replacement}
\end{align}
The ordinary law of large numbers gives
\begin{equation}
\widehat{\Cov}_{\mathcal D_k}\left\{
U_{\theta^\star}(X_i,Y_i)
-\frac{M_\mu\mu(X_i,Q_i)}{1+r}
\right\}
\xrightarrow{p}
\Cov\left\{
U_{\theta^\star}(X,Y)-\frac{M_\mu\mu(X,Q)}{1+r}
\right\}.
\label{eq:variance-proof-labeled-limit-covariance}
\end{equation}
Combining~\eqref{eq:variance-proof-labeled-covariance-replacement} and
\eqref{eq:variance-proof-labeled-limit-covariance}, and using
$\sum_{k=1}^3\omega_k=1$ and the fixed number of folds, proves
\eqref{eq:variance-proof-labeled-covariance}.
}

{We next prove
\eqref{eq:variance-proof-unlabeled-covariance}.} For the unlabeled
contribution, condition on the labeled data. {For
each fixed $k$ and every $\varepsilon>0$, conditional Markov's inequality
gives
\begin{align*}
&\Pr\!\left\{
\frac1N\sum_{j=n+1}^{n+N}
\|\widehat\mu^k(X_j,Q_j)-\mu(X_j,Q_j)\|^2>\varepsilon
\,\middle|\,\widehat\mu^k
\right\}\leq\frac1\varepsilon
\E\|\widehat\mu^k(X,Q)-\mu(X,Q)\|^2=o_p(1).
\end{align*}
Since the conditional probability is bounded by one,
\begin{align*}
&\Pr\!\left\{
\frac1N\sum_{j=n+1}^{n+N}
\|\widehat\mu^k(X_j,Q_j)-\mu(X_j,Q_j)\|^2>\varepsilon
\right\}\\
&\quad=
\E\!\left[
\Pr\!\left\{
\frac1N\sum_{j=n+1}^{n+N}
\|\widehat\mu^k(X_j,Q_j)-\mu(X_j,Q_j)\|^2>\varepsilon
\,\middle|\,\widehat\mu^k
\right\}
\right]\longrightarrow0.
\end{align*}
Combining this convergence with the ordinary law of large numbers and
$\E\|\mu(X,Q)\|^2<\infty$ gives
\begin{equation}
\frac1N\sum_{j=n+1}^{n+N}
\|\widehat\mu^k(X_j,Q_j)-\mu(X_j,Q_j)\|^2=o_p(1),
\qquad
\frac1N\sum_{j=n+1}^{n+N}\|\mu(X_j,Q_j)\|^2=O_p(1).
\label{eq:variance-proof-unlabeled-l2}
\end{equation}
Together with $\widehat M^k=M_\mu+o_p(1)$ and $n/N\to r$,
\eqref{eq:variance-proof-unlabeled-l2} gives
\begin{align}
&\frac1N\sum_{j=n+1}^{n+N}
\left\|
\frac{\widehat M^k\widehat\mu^k(X_j,Q_j)}{1+n/N}
-\frac{M_\mu\mu(X_j,Q_j)}{1+r}
\right\|^2=o_p(1).
\label{eq:variance-proof-foldwise-unlabeled-correction}
\end{align}
}
Since $\sum_{k=1}^3\omega_k=1$, convexity gives
\begin{align*}
&\frac1N\sum_{j=n+1}^{n+N}
\left\|
\frac1{1+n/N}\sum_{k=1}^3\omega_k
\widehat M^k\widehat\mu^k(X_j,Q_j)
-\frac1{1+r}M_\mu\mu(X_j,Q_j)
\right\|^2\\
&\quad\leq\sum_{k=1}^3\omega_k\frac1N\sum_{j=n+1}^{n+N}
\left\|
\frac{\widehat M^k\widehat\mu^k(X_j,Q_j)}{1+n/N}
-\frac{M_\mu\mu(X_j,Q_j)}{1+r}
\right\|^2.
\end{align*}
{By
\eqref{eq:variance-proof-foldwise-unlabeled-correction} and the fixed
number of folds, the right-hand side is $o_p(1)$. Hence}
\begin{equation}
\begin{split}
\frac1N\sum_{j=n+1}^{n+N}
\left\|
\frac1{1+n/N}\sum_{k=1}^3\omega_k
\widehat M^k\widehat\mu^k(X_j,Q_j)
-\frac1{1+r}M_\mu\mu(X_j,Q_j)
\right\|^2=o_p(1).
\end{split}
\label{eq:variance-proof-unlabeled-correction}
\end{equation}
{
The ordinary law of large numbers gives
\begin{equation}
\frac1N\sum_{j=n+1}^{n+N}
\left\|\frac1{1+r}M_\mu\mu(X_j,Q_j)\right\|^2=O_p(1).
\label{eq:variance-proof-unlabeled-limit-second-moment}
\end{equation}
The triangle inequality and
\eqref{eq:variance-proof-unlabeled-correction} give
\begin{align}
&\frac1N\sum_{j=n+1}^{n+N}
\left\|
\frac1{1+n/N}\sum_{k=1}^3\omega_k
\widehat M^k\widehat\mu^k(X_j,Q_j)
\right\|^2\\
&\quad\leq
\frac2N\sum_{j=n+1}^{n+N}
\left\|
\frac1{1+n/N}\sum_{k=1}^3\omega_k
\widehat M^k\widehat\mu^k(X_j,Q_j)
-\frac1{1+r}M_\mu\mu(X_j,Q_j)
\right\|^2\\
&\qquad+
\frac2N\sum_{j=n+1}^{n+N}
\left\|\frac1{1+r}M_\mu\mu(X_j,Q_j)\right\|^2
=O_p(1).
\label{eq:variance-proof-unlabeled-fitted-second-moment}
\end{align}
Applying the empirical Cauchy--Schwarz inequality to the empirical means
and second moments, using
\eqref{eq:variance-proof-unlabeled-correction},
\eqref{eq:variance-proof-unlabeled-limit-second-moment}, and
\eqref{eq:variance-proof-unlabeled-fitted-second-moment}, gives
\begin{align}
&\Bigg\|
\widehat{\Cov}_{\mathrm U}\left\{
\frac1{1+n/N}\sum_{k=1}^3\omega_k
\widehat M^k\widehat\mu^k(X_j,Q_j)
\right\}-
\widehat{\Cov}_{\mathrm U}\left\{
\frac1{1+r}M_\mu\mu(X_j,Q_j)
\right\}
\Bigg\|_F=o_p(1).
\label{eq:variance-proof-unlabeled-covariance-replacement}
\end{align}
The ordinary law of large numbers gives
\begin{equation}
\widehat{\Cov}_{\mathrm U}\left\{
\frac1{1+r}M_\mu\mu(X_j,Q_j)
\right\}
\xrightarrow{p}
\Cov\left\{\frac1{1+r}M_\mu\mu(X,Q)\right\}.
\label{eq:variance-proof-unlabeled-limit-covariance}
\end{equation}
Combining~\eqref{eq:variance-proof-unlabeled-covariance-replacement} and
\eqref{eq:variance-proof-unlabeled-limit-covariance} proves
\eqref{eq:variance-proof-unlabeled-covariance}.
}
The average over $k$ is taken before the covariance because all three estimators use the same unlabeled observations.

{
By the definition
\[
M_\mu=\Cov\{U_{\theta^\star}(X,Y),\mu(X,Q)\}
\Cov\{\mu(X,Q)\}^{-1},
\]
we have
\begin{align*}
\Cov\{U_{\theta^\star}(X,Y),M_\mu\mu(X,Q)\}
&=\Cov\{M_\mu\mu(X,Q),U_{\theta^\star}(X,Y)\}\\
&=\Cov\{M_\mu\mu(X,Q)\}\\
&=\Cov\{U_{\theta^\star}(X,Y),\mu(X,Q)\}
\Cov\{\mu(X,Q)\}^{-1}\\
&\quad\times\Cov\{\mu(X,Q),U_{\theta^\star}(X,Y)\}.
\end{align*}
Consequently,
\begin{align}
&\Cov\!\left\{
U_{\theta^\star}(X,Y)-\frac1{1+r}M_\mu\mu(X,Q)
\right\}
+r\Cov\!\left\{\frac1{1+r}M_\mu\mu(X,Q)\right\}\notag\\
&\quad=\Cov\{U_{\theta^\star}(X,Y)\}
-\frac1{1+r}
\Cov\{U_{\theta^\star}(X,Y),\mu(X,Q)\}
\Cov\{\mu(X,Q)\}^{-1}\notag\\
&\qquad\times\Cov\{\mu(X,Q),U_{\theta^\star}(X,Y)\}.
\label{eq:variance-proof-middle-matrix}
\end{align}
{
Combining~\eqref{eq:variance-proof-labeled-covariance},
\eqref{eq:variance-proof-unlabeled-covariance}, $n/N\to r$, and
\eqref{eq:variance-proof-middle-matrix} gives
\begin{align*}
&\sum_{k=1}^3\omega_k\widehat{\Cov}_{\mathcal D_k}
\left\{U_{\widehat\theta^{\mathrm{DiPPI}}}(X_i,Y_i)
-\frac{\widehat M^k\widehat\mu^k(X_i,Q_i)}{1+n/N}\right\}\\
&\quad+\frac nN\widehat{\Cov}_{\mathrm U}\left\{
\frac1{1+n/N}\sum_{k=1}^3\omega_k
\widehat M^k\widehat\mu^k(X_j,Q_j)\right\}\\
&\xrightarrow{p}
\Cov\{U_{\theta^\star}(X,Y)\}
-\frac1{1+r}
\Cov\{U_{\theta^\star}(X,Y),\mu(X,Q)\}
\Cov\{\mu(X,Q)\}^{-1}\\
&\qquad\times\Cov\{\mu(X,Q),U_{\theta^\star}(X,Y)\}.
\end{align*}
} Since $H_{\theta^\star}$ is nonsingular,
\[
\widehat H\xrightarrow{p}H_{\theta^\star}
\quad\Longrightarrow\quad
\widehat H^{-1}\xrightarrow{p}H_{\theta^\star}^{-1}.
\]
Therefore,
\[
\widehat\Sigma^{\mathrm{DiPPI}}
\xrightarrow{p}\Sigma_\mu^{\mathrm{DiPPI}},
\]
which proves~\eqref{eq:dippi-variance-consistency}.

For any fixed $a$ satisfying
$a^\top\Sigma_\mu^{\mathrm{DiPPI}}a>0$,
Theorem~\ref{thm:feasible-dippi} and
\eqref{eq:dippi-variance-consistency} give
\begin{align*}
\frac{\sqrt n\,a^\top
(\widehat\theta^{\mathrm{DiPPI}}-\theta^\star)}
{\{a^\top\Sigma_\mu^{\mathrm{DiPPI}}a\}^{1/2}}
&\xrightarrow{d}N(0,1),\\
\frac{a^\top\widehat\Sigma^{\mathrm{DiPPI}}a}
{a^\top\Sigma_\mu^{\mathrm{DiPPI}}a}
&\xrightarrow{p}1.
\end{align*}
In particular,
\[
\Pr\{a^\top\widehat\Sigma^{\mathrm{DiPPI}}a>0\}\longrightarrow1,
\qquad
\left\{
\frac{a^\top\widehat\Sigma^{\mathrm{DiPPI}}a}
{a^\top\Sigma_\mu^{\mathrm{DiPPI}}a}
\right\}^{1/2}\xrightarrow{p}1.
\]
Slutsky's theorem therefore proves~\eqref{eq:dippi-studentized-clt}.
Finally,
\begin{align*}
&\Pr\!\left\{
a^\top\theta^\star\in
a^\top\widehat\theta^{\mathrm{DiPPI}}
\ \pm\ z_{1-\alpha/2}
\left(\frac{a^\top\widehat\Sigma^{\mathrm{DiPPI}}a}{n}\right)^{1/2}
\right\}\\
&\quad=
\Pr\!\left\{
\left|
\frac{\sqrt n\,a^\top
(\widehat\theta^{\mathrm{DiPPI}}-\theta^\star)}
{\{a^\top\widehat\Sigma^{\mathrm{DiPPI}}a\}^{1/2}}
\right|\leq z_{1-\alpha/2}
\right\}
\longrightarrow1-\alpha,
\end{align*}
which proves the stated coverage. %
}
\end{proof}

%

\bibliographystyle{plainnat}
\bibliography{references}

\end{document}